\documentclass[sigconf, nonacm]{acmart}

\usepackage{bm}
\usepackage{caption}
\usepackage{subcaption}
\usepackage[ruled, vlined, linesnumbered]{algorithm2e}
\usepackage{textcomp}
\usepackage{enumerate}
\usepackage{enumitem}
\usepackage{layouts}
\usepackage{multirow}
\usepackage{soul}
\usepackage{siunitx}
\usepackage{balance}

\newtheorem{prob}{\textsc{Problem}}
\newtheorem{subprob}{\textsc{Problem}}[prob]
\newtheorem{lem}{Lemma}

\newcommand{\fd}{\textsc{FrequentDirections}\xspace}

\newcommand{\swfd}{\textsc{FrequentDirections} over sliding window\xspace}
\newcommand{\lmfd}{\textsf{LM-FD}\xspace}
\newcommand{\difd}{\textsf{DI-FD}\xspace}
\newcommand{\dsfd}{\textsf{DS-FD}\xspace}
\newcommand{\fastdsfd}{\textsf{Fast-DS-FD}\xspace}
\newcommand{\seqdsfd}{\textsf{Seq-DS-FD}\xspace}

\def\header{\noindent}

\newcommand\vldbdoi{10.14778/3665844.3665847}
\newcommand\vldbpages{2149 - 2161}
\newcommand\vldbvolume{17}
\newcommand\vldbissue{9}
\newcommand\vldbyear{2024}
\newcommand\vldbauthors{\authors}
\newcommand\vldbtitle{\shorttitle} 
\newcommand\vldbavailabilityurl{https://github.com/yinhanyan/DS-FD}
\newcommand\vldbpagestyle{empty} 

\begin{document}
\title{Optimal Matrix Sketching over Sliding Windows}


\author{Hanyan Yin}
\affiliation{%
  \institution{Renmin University of China}
}
\email{yinhanyan@ruc.edu.cn}

\author{Dongxie Wen}
\affiliation{%
  \institution{Renmin University of China}
}
\email{2019202221@ruc.edu.cn}

\author{Jiajun Li}
\affiliation{%
  \institution{Renmin University of China}
}
\email{2015201613@ruc.edu.cn}

\author{Zhewei Wei}
\authornote{Zhewei Wei is the corresponding author. The work was partially done at Gaoling School of Artificial Intelligence, Beijing Key Laboratory of Big Data Management and Analysis Methods, MOE Key Lab of Data Engineering and Knowledge Engineering, and Pazhou Laboratory (Huangpu), Guangzhou, Guangdong 510555, China. }
\affiliation{%
  \institution{Renmin University of China}
}
\email{zhewei@ruc.edu.cn}

\author{Xiao Zhang}
\affiliation{%
  \institution{Renmin University of China}
}
\email{zhangx89@ruc.edu.cn}

\author{Zengfeng Huang}
\affiliation{%
  \institution{Fudan University}
}
\email{huangzf@fudan.edu.cn}

\author{Feifei Li}
\affiliation{%
  \institution{Alibaba Group}
}
\email{lifeifei@alibaba-inc.com}







\begin{abstract}

Matrix sketching, aimed at approximating a matrix $\bm{A} \in \mathbb{R}^{N\times d}$ consisting of vector streams of length $N$ with a smaller sketching matrix $\bm{B} \in \mathbb{R}^{\ell\times d}, \ell \ll N$, has garnered increasing attention in fields such as large-scale data analytics and machine learning. A well-known deterministic matrix sketching method is the \fd algorithm, which achieves the optimal $O\left(\frac{d}{\varepsilon}\right)$ space bound and provides a covariance error guarantee of $\varepsilon = \lVert \bm{A}^\top \bm{A} - \bm{B}^\top \bm{B} \rVert_2/\lVert \bm{A} \rVert_F^2$. The matrix sketching problem becomes particularly interesting in the context of sliding windows, where the goal is to approximate the matrix $\bm{A}_W$, formed by input vectors over the most recent $N$ time units. However, despite recent efforts, whether achieving the optimal $O\left(\frac{d}{\varepsilon}\right)$ space bound on sliding windows is possible has remained an open question.

In this paper, we introduce the \dsfd algorithm, which achieves the optimal $O\left(\frac{d}{\varepsilon}\right)$ space bound for matrix sketching over row-normalized, sequence-based sliding windows. We also present matching upper and lower space bounds for time-based and unnormalized sliding windows, demonstrating the generality and optimality of \dsfd across various sliding window models. This conclusively answers the open question regarding the optimal space bound for matrix sketching over sliding windows. We conduct extensive experiments with both synthetic and real-world datasets, validating our theoretical claims and thus confirming the correctness and effectiveness of our algorithm, both theoretically and empirically.

\end{abstract}

\maketitle

\pagestyle{\vldbpagestyle}
\begingroup\small\noindent\raggedright\textbf{PVLDB Reference Format:}\\
\vldbauthors. \vldbtitle. PVLDB, \vldbvolume(\vldbissue): \vldbpages, \vldbyear.\\
\href{https://doi.org/\vldbdoi}{doi:\vldbdoi}
\endgroup
\begingroup
\renewcommand\thefootnote{}\footnote{\noindent
This work is licensed under the Creative Commons BY-NC-ND 4.0 International License. Visit \url{https://creativecommons.org/licenses/by-nc-nd/4.0/} to view a copy of this license. For any use beyond those covered by this license, obtain permission by emailing \href{mailto:info@vldb.org}{info@vldb.org}. Copyright is held by the owner/author(s). Publication rights licensed to the VLDB Endowment. \\
\raggedright Proceedings of the VLDB Endowment, Vol. \vldbvolume, No. \vldbissue\ %
ISSN 2150-8097. \\
\href{https://doi.org/\vldbdoi}{doi:\vldbdoi} \\
}\addtocounter{footnote}{-1}\endgroup

\ifdefempty{\vldbavailabilityurl}{}{
\vspace{.3cm}
\begingroup\small\noindent\raggedright\textbf{PVLDB Artifact Availability:}\\
The source code, data, and/or other artifacts have been made available at \url{https://github.com/yinhanyan/DS-FD}.
\endgroup
}

\section{Introduction}

\begin{table*}[ht]
\caption{Given the dimension \(d\) of each row vector, the upper bound of relative covariance error \(\varepsilon\), and the size of the sliding window \(N\), this table presents an overview of space complexities for algorithms addressing matrix sketching over sliding windows. An asterisk (*) indicates that the space complexity is the expected value when it is a random variable. For each column, \textit{sequence-based} denotes that each update occupies a timestamp, \textit{time-based} denotes that each timestamp unit may contain zero or multiple updates, \textit{normalized} denotes the norm of each row equals a constant, and \textit{unnormalized} denotes the norm of each row \(\lVert \bm{a}_i\rVert_2^2 \in [1, R]\).
}
\label{tab:alg}
\begin{tabular}{l|l|l|l|l}
\toprule
\multirow{2}{*}{sketch $\kappa$} & \multicolumn{2}{c|}{Sequence-based} & \multicolumn{2}{c}{Time-based} \\ \cline{2-5}
& normalized & unnormalized & normalized & unnormalized \\ \hline
Sampling(SWR)~\cite{braverman2020near,wei2016matrix} & $O\left({d\over \varepsilon^2} \log N\right)$ * & $O\left({d\over \varepsilon^2}\log NR\right)$ * & $O\left({d\over \varepsilon^2}\log N\right)$ * & $O\left({d\over \varepsilon^2}\log NR\right)$ *\\
LM-HASH~\cite{clarkson2017low} & $O\left({d^2\over\varepsilon^3}\right)$ & $O\left({d^2\over \varepsilon^3}\log R\right)$ & $O\left({d^2\over \varepsilon^3}\log \varepsilon N\right)$ & $O\left({d^2\over \varepsilon^3}\log \varepsilon NR\right)$ \\
DI-RP~\cite{boutsidis2014near} & $O\left({d\over \varepsilon^2}\log {1\over \varepsilon} \right)$ & $O\left({Rd \over \varepsilon^2}\log {R\over \varepsilon}\right)$ & - & - \\
DI-HASH~\cite{clarkson2017low} & $O\left({d^2\over \varepsilon^2}\log {1\over \varepsilon}\right)$ & $O\left({Rd^2\over \varepsilon^2}\log {R\over \varepsilon}\right)$ & - & - \\
\lmfd~\cite{liberty2013simple,wei2016matrix} & $O\left({d\over \varepsilon^2}\right)$ & $O\left({d \over \varepsilon^2}\log R\right)$ & $O\left({d \over \varepsilon^2}\log \varepsilon N\right)$ & $O\left({d \over \varepsilon^2}\log \varepsilon NR\right)$  \\
\difd~\cite{liberty2013simple,wei2016matrix} & $O\left({d \over \varepsilon}\log {1 \over \varepsilon}\right)$ & $O\left({Rd \over \varepsilon}\log {R \over \varepsilon}\right)$ & - & - \\ \hline
\dsfd (This paper) & $O\left(d \over \varepsilon\right)$ & $O\left({d \over \varepsilon} \log R\right)$ & $O\left({d \over \varepsilon} \log \varepsilon N\right)$ & $O\left({d \over \varepsilon} \log \varepsilon NR\right)$ \\
Lower bound (This paper) & $\Omega\left({d\over \varepsilon}\right)$ & $\Omega\left({d \over \varepsilon} \log R\right)$ & $\Omega\left({d \over \varepsilon} \log \varepsilon N\right)$ & $\Omega\left({d \over \varepsilon} \log \varepsilon NR\right)$ \\ 
\bottomrule
\end{tabular}
\end{table*}

Many types of real-world streaming data, such as computer networking traffic, social media content, and sensor data, are continuously generated, often arriving in large volumes or at high speeds~\cite{gaber2005mining,zhou2023stream}. Given the constraints in storage and computational resources, it becomes frequently impractical to store or compute aggregations and statistics for streaming data accurately. Algorithms for streaming data provide approximate solutions by summarizing, sketching, or synthesizing the data stream with sublinear space or time complexity relative to the input size~\cite{muthukrishnan2005data,zeng2022persistent}. Among various streaming data algorithms, matrix sketching emerges as a general technique designed to process streaming data comprised of vectors or matrices~\cite{woodruff2014sketching}. A wide array of matrix sketching algorithms has been proposed, categorized into several approaches: sparsification~\cite{arora2006fast,drineas2011note}, sampling~\cite{rudelson2007sampling,deshpande2006adaptive,arora2006fast}, random projection~\cite{sarlos2006improved,vempala2005random}, hashing~\cite{weinberger2009feature,clarkson2017low}, and \fd(\textsf{FD})~\cite{liberty2013simple,ghashami2016frequent}. These algorithms typically present trade-offs between time-space complexity and accuracy. Notably, \fd~\cite{liberty2013simple,ghashami2016frequent}, a prominent deterministic algorithm, achieves a spectral bound on the relative covariance error, expressed as \(\varepsilon = \lVert \bm{A}^\top \bm{A} - \bm{B}^\top \bm{B}\rVert_2/\lVert\bm{A}\rVert_F^2\), with an optimal space complexity of \(O\left(\frac{d}{\varepsilon}\right)\). These attributes have led to its widespread application across various fields~\cite{luo2019robust,feinberg2024sketchy,dickens2020ridge,chen2021efficient,chen2022efficient}.

In real-world scenarios, the interest often lies in the most recently arrived elements rather than outdated items within data streams, as highlighted by recent studies~\cite{wei2016matrix,zhang2017tracking,cormode2020small}. Datar et al.~\cite{datar2002maintaining} consider the problem of maintaining aggregates and statistics from the most recent period of the data stream and refer to such a model as the \textit{sliding window model}. This paper delves into the \textit{continuous tracking matrix sketch over sliding windows}, a crucial technique for applications like \textsf{sliding window PCA} or \textsf{real-time PCA}~\cite{chowdhury2020real,wei2016matrix}. Such techniques play a vital role across various domains, including event detection~\cite{rafferty2016real}, fault monitoring~\cite{sheriff2017fault,ammiche2018modified}, differential privacy~\cite{upadhyay2021framework} and online learning~\cite{garivier2011upper}, highlighting the significance of optimizing both the complexity and the quality of estimations provided by matrix sketching algorithms for sliding windows.

Over the years, extensive research has been dedicated to developing improved sketching algorithms to address the challenge of matrix sketching over sliding windows. For instance, Wei et al.~\cite{wei2016matrix} proposed the Sampling, \lmfd, and \difd algorithms. Sampling is a probabilistic algorithm, and \lmfd and \difd are deterministic algorithms that build upon \fd. In addition, streaming matrix sketching algorithms based on \textit{random projection}~\cite{boutsidis2014near} and \textit{hashing}~\cite{clarkson2017low}, such as \textsf{LM-HASH}, \textsf{DI-RP}, and \textsf{DI-HASH}, are also compatible with the sliding window model. Zhang et al.~\cite{zhang2017tracking} explored the challenge of tracking matrix approximations over distributed sliding windows, proposing communication-efficient algorithms like priority sampling, ES sampling, and DA. Braverman et al.~\cite{braverman2020near} present a randomized row sampling framework for a wide spectrum of linear algebra approximation problems and a unified framework for deterministic algorithms in the sliding window model based on the merge-and-reduce paradigm and the online coresets. Furthermore, Shi et al.~\cite{shi2021time} also present a method to extend \fd to the persistent summary model, another historical range query model similar to the sliding window. Table~\ref{tab:alg} compares the memory cost of various matrix sketching algorithms over sliding windows.

\header{\bf Motivations.} Despite the significant efforts mentioned above, existing algorithms for matrix sketching over sliding windows remain sub-optimal in terms of space complexity. Table~\ref{tab:alg} illustrates that for the fundamental case where the window is sequence-based (i.e., each update occupies a timestamp) and each row is normalized (i.e., the norm of each row equals a constant), the current best space complexity for matrix sketching algorithms over sliding windows is \(O\left(\frac{d}{\varepsilon}\log \frac{1}{\varepsilon} \right)\) for \difd and \(O\left(\frac{d}{\varepsilon^2}\right)\) for \lmfd~\cite{wei2016matrix}. Conversely, it has been demonstrated in~\cite{ghashami2016frequent} that the space lower bound for \fd in the full stream model is \(\Omega\left(\frac{d}{\varepsilon}\right)\), thus establishing the optimality of \fd. Given that the sliding window model poses greater challenges than the full stream model, \(\Omega\left(\frac{d}{\varepsilon}\right)\) also represents a space lower bound for the sliding window model. This observation leads to a compelling inquiry: \textbf{Is it possible to achieve the optimal \(O\left(\frac{d}{\varepsilon}\right)\) space bound for the problem of matrix sketching over the sliding window model?}

To address the question, it is crucial to recognize that the existing methods fail to achieve optimal space complexity primarily because they merely integrate \fd with generic sliding window algorithmic frameworks. This approach lacks a profound examination and enhancement of the original \fd algorithm specifically tailored for sliding windows. For instance, \lmfd applies \fd within the Exponential Histogram (EH) framework~\cite{datar2002maintaining}, and \difd combines \fd with the Dyadic Interval (DI) framework~\cite{wei2016matrix}. While these well-known frameworks are adept at adapting various streaming algorithms, such as Misra-Gries and SpaceSaving, to the sliding window model, they inherently introduce a multiplicative increase in memory overhead~\cite{lee2006simpler}. To develop more space-efficient sliding window algorithms, it is often necessary to undertake modifications directly on the streaming algorithms. This paper takes such an approach with \fd, aiming to refine and optimize it specifically for the sliding window context.

\subsection{Our Contributions}
In this paper, we introduce \dsfd, a deterministic algorithm that achieves optimal space complexity for matrix sketching over sliding windows. Specifically, \dsfd reaches an optimal space complexity of \(O\left(\frac{d}{\varepsilon}\right)\) for sequence-based sliding windows with normalized rows. When the norms of the rows are not normalized and falling within the range of \([1, R]\), the space complexity naturally expands to \(O\left(\frac{d}{\varepsilon}\log R\right)\). For time-based sliding windows, where each timestamp may not correspond to a row update, the space complexities are adjusted to \(O\left(\frac{d}{\varepsilon}\log \varepsilon N\right)\) (row-normalized) and \(O\left(\frac{d}{\varepsilon}\log \varepsilon N R\right)\) (row-unnormalized), respectively. These space complexities are detailed in the second-to-last row of Table~\ref{tab:alg}.

Furthermore, we establish the lower bound of space complexity for any deterministic matrix sketching algorithm over sliding windows. Surprisingly, the corresponding lower bounds for the four models are also \(\Omega\left(\frac{d}{\varepsilon}\right)\), \(\Omega\left(\frac{d}{\varepsilon}\log R\right)\), \(\Omega\left(\frac{d}{\varepsilon}\log \varepsilon N\right)\), and \(\Omega\left(\frac{d}{\varepsilon}\log \varepsilon NR\right)\), as detailed in the last row of Table~\ref{tab:alg}. This signifies that we have achieved optimal space complexity across the four distinct models. The specific contributions of this paper are outlined below:

\begin{itemize}[leftmargin = *]
    \item \textbf{Novel Algorithm with Improved Complexity:} We theoretically demonstrate the superiority of our \dsfd algorithm over existing methods in matrix sketching over sliding windows. Specifically, in the second-last column of Table~\ref{tab:alg}, we highlight the theoretical improvements of \dsfd compared to existing methods. The space complexity \(O\left(\frac{d}{\varepsilon}\right)\) of \dsfd is more efficient than that of two leading algorithms, \(O\left(\frac{d}{\varepsilon^2}\right)\) for \lmfd and \(O\left(\frac{d}{\varepsilon} \log \frac{1}{\varepsilon}\right)\) for \difd. More importantly, our \dsfd algorithm implements modifications on the \fd based on our deeper insight into its application in sliding windows, rather than merely incorporating \fd into a generic sliding window framework, as done by \lmfd and \difd. Our \dsfd offers a deterministic error bound and applies to both sequence-based and time-based windows, achieving an amortized update time of \(O\left(d\ell\right)\), which is on par with the fastest FD in the full stream setting. 

    
    \item \textbf{Matching Lower Bounds:} We establish matching lower bounds for any deterministic matrix sketching algorithm over sliding windows. Our proof is inspired by techniques used in the space lower bound proofs of \textsc{BasicCounting}~\cite{datar2002maintaining} and other streaming matrix sketching problems~\cite{ghashami2016frequent}. This validation confirms that our \dsfd algorithm is optimal in terms of space complexity.


    \item \textbf{Extensive Experiments:} We conduct comprehensive experimental studies to verify the superiority of \dsfd over other state-of-the-art algorithms, especially in terms of sketch memory usage. Our experimental results reveal that the trade-off between error bounds and space cost for the \dsfd algorithm is more favorable than existing algorithms on synthetic and real-world datasets. Furthermore, optimizing space cost becomes increasingly significant as the permissible upper bound of covariance relative error tightens. These experimental findings are in strong agreement with our theoretical analyses.
\end{itemize}
\section{Preliminaries}

This section introduces widely-used problem definitions of matrix sketching over sliding windows and several fundamental concepts related to this topic.

\subsection{Problem Definition}

\begin{prob}[Matrix Sketching over Sliding Window]
\label{prob:sw-fd}
Suppose we have a data stream where each item is in the set $\mathbb{R}^{d}$. Given the error parameter $\varepsilon$ and window size $N$, the goal is to maintain a matrix sketch $\kappa$ such that, at the current time $T$, $\kappa$ can return an approximation $\bm{B}_W$ for the matrix $\bm{A}_W=\bm{A}_{T-N,T} \in \mathbb{R}^{N\times d}$, stacked by the recent $N$ items. The approximation quality is measured by the \textit{covariance error}, such that:
\begin{equation*}
    \textup{\textbf{cova-error}}(\bm{A}_W, \bm{B}_W) = \lVert \bm{A}^\top_W\bm{A}_W - \bm{B}^\top_W\bm{B}_W\rVert_2 \le \varepsilon \lVert\bm{A}_W\rVert_F^2,
\end{equation*}
where $N$ bounds the maximum size of window $W$.
\end{prob}

Wei et al.~\cite{wei2016matrix} show that maintaining a deterministic sketch with space less than $o(Nd)$ is not feasible if the maximum norm of row vectors of $\bm{A}_W$ is unbounded, even when a large covariance error is permissible. Consequently, contemporary research concentrates on problem variants where the maximum possible squared norm of all rows is upper-bounded. We define the following four variants of the \swfd problem. The paper will sequentially introduce algorithms to tackle these four problems, starting with the simplest, the sequence-based normalized window model, to illustrate the concept of dump snapshots. This approach serves as a foundation for addressing the unnormalized or the time-based model. We begin by formulating the problem for the sequence-based normalized window model, which is the simplest to solve and analyze.

\begin{subprob}[Matrix Sketching over Sequence-based Normalized Sliding Window]
\label{prob:norm-seq-sw-fd}
We assume the squared norms of the rows in the window take value 1, that is, $\lVert \bm{a}\rVert^2_2=1$ for all $\bm{a}\in W$. Therefore, the \textit{covariance error} is,
$$
\textup{\textbf{cova-error}}(\bm{A}_W, \bm{B}_W) = \lVert \bm{A}^\top_W\bm{A}_W - \bm{B}^\top_W\bm{B}_W\rVert_2 \le \varepsilon \lVert\bm{A}_W\rVert_F^2 = \varepsilon N.
$$
\end{subprob}

Next, we eliminate the normalization restriction, resulting in the following variant of the problem definition.

\begin{subprob}[Matrix Sketching over Sequence-based Sliding Window]
\label{prob:seq-sw-fd}
We assume the squared norms of the rows in the window range from $[1, R]$, that is, $1\le \lVert \bm{a}\rVert^2_2\le R$ for all $\bm{a}\in W$.
\end{subprob}

If we account for real-world time instead of sequential order, the problems of both the normalized and unnormalized versions can be stated as follows.

\begin{subprob}[Matrix Sketching over Time-based Normalized Sliding Window]
\label{prob:norm-time-sw-fd}
We assume the squared norms of the rows in the window are either $0$ or $1$, that is, $\lVert \bm{a}\rVert^2_2 = 0$ or $\lVert \bm{a}\rVert^2_2 = 1$ for all $\bm{a}\in W$.
\end{subprob}

\begin{subprob}[Matrix Sketching over Time-based Sliding Window]
\label{prob:time-sw-fd}
We assume the squared norms of the rows in the window range from $\{0\}\cup [1,R]$, that is, $\lVert \bm{a}\rVert^2_2=0$ or $1\le \lVert \bm{a}\rVert^2_2\le R$ for all $\bm{a}\in W$.
\end{subprob}

\subsection{\fd}

\fd~\cite{liberty2013simple,ghashami2016frequent} is a deterministic algorithm of matrix sketching in the row-update model. It processes one row vector \(\bm{a}_i \in \mathbb{R}^{1 \times d}\) at a time, accumulating a stream of row vectors to construct a matrix \(\bm{A} \in \mathbb{R}^{n \times d}\), where \(d \ll n\). Let $\ell=1/\varepsilon$, \fd takes amortized \(O(d\ell)\) time per row and maintains a sketch matrix \(\bm{B} \in \mathbb{R}^{\ell \times d}\), achieving an error bound of
\begin{equation*}
    \textbf{cova-error}(\bm{A}, \bm{B}) = \lVert \bm{A}^\top\bm{A} - \bm{B}^\top\bm{B}\rVert_2 \le \varepsilon \lVert\bm{A}\rVert_F^2.
\end{equation*}

To process each arriving row vector, we first check whether \(\bm{B}\) has any zero-valued rows. If \(\bm{B}\) contains a zero-valued row, we insert \(\bm{a}_i\) into it. Otherwise, we perform a Singular Value Decomposition (SVD) \([\bm{U},\bm{\Sigma}, \bm{V}^\top] = \texttt{svd}(\bm{B})\), rescale the "directions" in \(\bm{B}\) with the \(\ell\)-th largest singular value \(\sigma_\ell\), and "forget" the least significant direction in the column space of \(\bm{V}^\top\). The updated \(\bm{\Sigma}^\prime\) and sketch matrix \(\bm{B}^\prime\) are computed as follows:
\begin{equation*}
    \Sigma^\prime = \text{diag}\left(\sqrt{\sigma_1^2-\sigma_\ell^2}, \dots, \sqrt{\sigma_{\ell-1}^2-\sigma_\ell^2}, 0, \dots, 0 \right)
\end{equation*}
and \(\bm{B}^\prime = \bm{\Sigma}^\prime \bm{V}^\top\), where \(\sigma_1^2 \ge \sigma_2^2 \ge \dots \ge \sigma_{2\ell}^2\). The updated \(\bm{B}^\prime\) will have at least \(\ell+1\) nonzero rows, allowing the process to continuously accommodate the next arriving row vector and update the sketch matrix \(\bm{B}\) in the same manner as described above.

\fd can also be adapted for the sliding window model. We briefly introduce two algorithms based on it, which are the main competitors to our novel algorithms:
\begin{itemize}[leftmargin = *]
    \item \lmfd applies Exponential Histogram~\cite{datar2002maintaining} to \fd. It organizes the window into exponentially shrinking block sizes, with the most recent level having blocks of size $\ell$ and the oldest level being $\varepsilon \lVert \bm{A}_W \rVert_F^2$. As shown in Table~\ref{tab:alg}, it uses $O\left(\frac{d}{\varepsilon^2}\right)$ space for sequence-based normalized sliding windows. 
    \item \difd adopts the Dyadic Interval~\cite{arasu2004approximate} approach to \fd, maintaining $L=\log \frac{R}{\varepsilon}$ parallel levels. At the $j$-th level, the matrix over the sliding window is segmented into sizes of $\Theta\left(\frac{NR}{2^{L-i}}\right)$. \difd uses $O\left(\frac{d}{\varepsilon}\log \frac{1}{\varepsilon}\right)$ space for sequence-based normalized sliding windows.
\end{itemize}

\section{Our method}

In this section, we present a fundamental method that addresses Problem~\ref{prob:norm-seq-sw-fd}, termed \textbf{Dump Snapshots (\textsf{DS})}. This method draws inspiration from the $\lambda$-snapshot method employed for solving the $\varepsilon$-approximate frequent items problem over the sliding window~\cite{lee2006simpler} and its connection to matrix sketching over the sliding window~\cite{ghashami2016frequent}. 
The algorithms developed for solving both the unnormalized and time-based problems are also based on this approach.


\header{\textbf{Connection to Item Frequency Estimation.}} The matrix sketching problem over sliding windows has a significant relation to the \textit{$\varepsilon$-approximation frequent items} problem over sliding windows. The definition of this latter problem is outlined in~\cite{zhang2008frequency,lee2006simpler}:

\begin{prob}[$\varepsilon$-approximation Frequent Items over Sliding Window]
\label{prob:swfe}
Consider a data stream where each item belongs to the set $[d]$. Given an error parameter $\varepsilon$ and a window size $N$, the goal is to maintain a sketch $\kappa$ such that that ensures, for any item $i \in [d]$, the error between its true frequency $f_i$ and estimated frequency $\hat{f}_i$ over the most recent $N$ items returned by $\kappa$ is constrained by:
\begin{equation*}
    \left| f_i-\hat{f}_i\right| \le \varepsilon N.
\end{equation*}
\end{prob}

For the Problem~\ref{prob:swfe}, each element arriving in the data stream $\bm{A}_W=\{\bm{a}_1,\bm{a}_2,\dots,\bm{a}_n\}$ can be viewed as a one-hot indicator vector, denoted as $\bm{a}_i \in \{\bm{e}_1, ..., \bm{e}_d\}$, where $\bm{e}_j$ represents the $j$-th standard basis vector. The precise frequency of item $j$ can be expressed as $f_j=\lVert \bm{A}_W \bm{e}_j \rVert_2^2$, which corresponds to the $j$-th element on the diagonal of matrix $\bm{A}_W^\top \bm{A}_W$. Assuming an algorithm that offers an estimated frequency $\hat{f_i}$ for each item $i$, we designate the $i$-th row of matrix $\bm{B}$ as $\hat{f}_i^{1/2}\cdot \bm{e}_i$. Consequently, the error bound for the Problem~\ref{prob:swfe}, $\max_i \left| f_i-\hat{f_i}\right| \le \varepsilon N$, can be expressed as $\lVert \bm{A}^\top_W\bm{A}_W - \bm{B}^\top_W\bm{B}_W\rVert_2 \le \varepsilon N$, aligning with Problem~\ref{prob:norm-seq-sw-fd}. Hence, the Problem~\ref{prob:swfe} over sliding window can be regarded as a special case of matrix sketching problem outlined in Problem~\ref{prob:norm-seq-sw-fd}, where the incoming row vectors are mutually orthogonal.

The Misra-Gries (MG) summary is a well-known algorithm for addressing the \textit{$\varepsilon$-approximation frequent items} problem~\cite{misra1982finding}. It operates by maintaining $1/\varepsilon$ counters, each initially set to 0. Upon encountering a new item, the algorithm increments its corresponding counter if it already has one; otherwise, it assigns a counter to it, setting its value to 1. If, following this increment, all $1/\varepsilon$ counters possess a value greater than 0, the algorithm decrements each counter by one, ensuring that there will be at least one counter reset to 0. This method effectively keeps track of the most frequent items within a data stream. 

\header{\bf Extension of MG summary to sliding window.} The MG summary can be effectively adapted to the sliding window model by the $\lambda$-snapshot method introduced by Lee et al.~\cite{lee2006simpler}. It preserves the error-bound guarantee without additional space or time complexity. This method's core concept is to log the event "item $i$ has appeared $\varepsilon N$ times" along with the current timestamp and expire these logs in time. Similar to the original MG summary, the $\lambda$-snapshot method demands $O(1/\varepsilon)$ time for both updates and queries, and requires $O(1/\varepsilon)$ space. Moreover, an optimization has been proposed that can further diminish the time complexity for updates and queries to $O(1)$, while keeping the space complexity unchanged~\cite{zhang2008frequency}.

\begin{figure}[h]
    \centering
    \includegraphics[width=\linewidth]{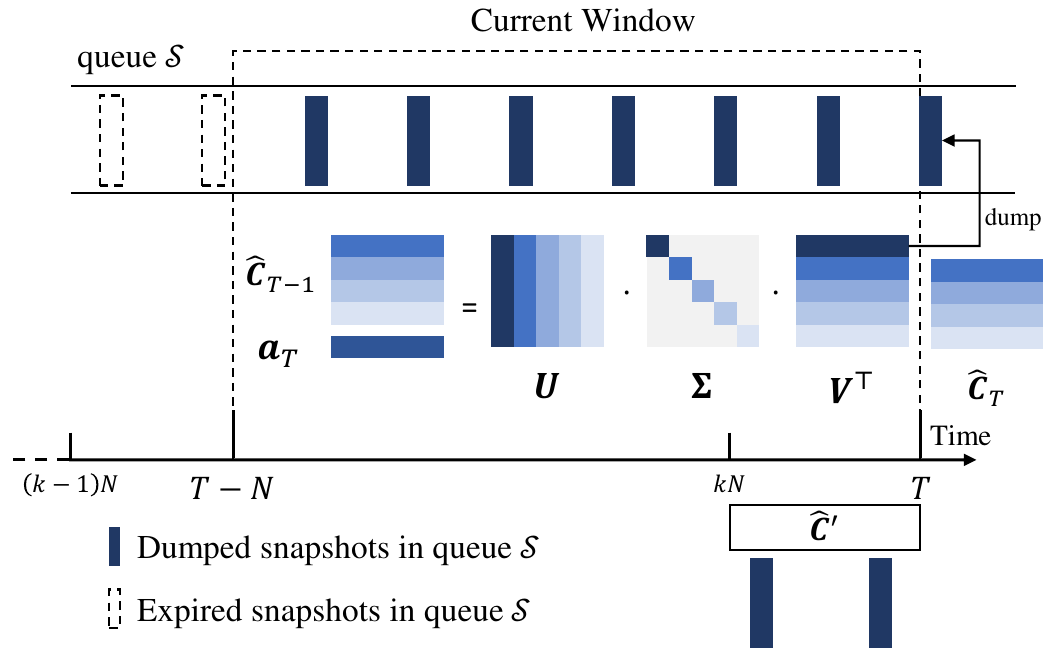}
    \caption{The data structures and update steps of \dsfd entail performing an SVD decomposition $\bm{U\Sigma V^\top} = \texttt{svd}([\hat{\bm{C}}_{T-1}, \bm{a}_T])$ for each update. Following the decomposition, singular values and their corresponding right singular vectors are evaluated against the error bound $\varepsilon N$. Those exceeding the bound are "dumped," i.e., removed from the current sketch and stored as snapshots in a queue $\mathcal{S}$, accompanied by the current timestamp.}
    \label{fig:dsfd}
\end{figure}

Given the connection outlined above, where both the FD sketch and the $\lambda$-snapshot method derive from the principles of the MG summary, a compelling question emerges: Is it possible to extend the FD matrix sketch to the sliding window model while maintaining the same error-bound guarantee and the same space and time complexity as done by the $\lambda$-snapshot method? This question paves the way for the introduction of the \dsfd method.

\header{\bf Extension of FD Summary to Sliding Window.} The core concept of \dsfd is illustrated in Figure~\ref{fig:dsfd}: for each FD update, after performing an SVD $\left[\bm{U},\bm{\Sigma},\bm{V}^\top\right] = \texttt{svd}(\left[\hat{\bm{C}}_{T-1}, \bm{a}_T\right])$, we "dump" the singular values and their corresponding right singular vectors if the singular values exceed the error threshold $\varepsilon N$. These vectors are stored as snapshots in a queue $\mathcal{S}$ with the timestamps. When a query is invoked, we combine the vectors from the snapshots within the current window's timeframe with the FD sketch to produce a matrix $\bm{B}_W$, which is then provided as the sketch.

Despite the simplicity of this approach, proving its correctness and achieving efficient implementation pose significant challenges. The main difficulties are: (1) The change in the orthogonal basis formed by the right singular vectors from the SVD of $\bm{A}_W$ before and after a single-step update, which results in past snapshots being non-orthogonal with current right singular vectors, complicating the error analysis. We prove the error bound for our algorithm as Theorem~\ref{thm:norm-seq-sw-fd}. (2) The challenge of accurately recording the timestamps of dumped snapshot vectors, which complicates the direct application of the FAST-FD algorithm. We also try to optimize it and propose the \fastdsfd. 

The space requirement for \dsfd is $O(d/\varepsilon)$, aligning with the space lower bound of FD, thus confirming the optimality of our algorithm in terms of space usage. Additionally, we plan to extend our algorithms to address more general scenarios, including arriving vectors with norms $\lVert \bm{a}_i\rVert_2^2 \in [1, R]$, in Sections~\ref{sec:seq-dsfd} and~\ref{sec:time-dsfd}.

\subsection{Algorithm Description}

\begin{algorithm}[h]
	\caption{ \dsfd: \textsc{Initialize}($d, \ell, N, \theta$)}
    \label{alg:norm-seq-sw-fd-init}
	\KwIn{ 
        $d$: Dimension of input vectors of \textsc{Update}\\
        $\ell=\min\left(\lceil\frac{1}{\varepsilon}\rceil, d\right)$: Number of rows in FD sketch \\
        $N$: Length of sliding window\\
        $\theta$: Dump threshold. For \textsc{Problem} \ref{prob:norm-seq-sw-fd}, $\theta=\varepsilon N$ \\
    }
    $\hat{\bm{C}}\leftarrow \bm{0}_{\ell\times d}$ \\
    $\hat{\bm{C}}^\prime \leftarrow \bm{0}_{\ell\times d}$ \\
    Queue of snapshots $\mathcal{S}\leftarrow \mathsf{queue}.\textsc{Initialize()}$\\
    Auxiliary queue of snapshots $\mathcal{S}^\prime \leftarrow \mathsf{queue}.\textsc{Initialize()}$ \\
\end{algorithm}

\header{\bf Data structures.} Figure~\ref{fig:dsfd} and Algorithm~\ref{alg:norm-seq-sw-fd-init} show the data structures of \dsfd. At any given moment, \dsfd maintains two FD sketches: a primary sketch $\hat{\bm{C}}$ and an auxiliary sketch $\hat{\bm{C}}^\prime$. Each sketch is associated with its corresponding queue of snapshots, $\mathcal{S}$ and $\mathcal{S}^\prime$, respectively. By setting $\ell=\min\left(d, \lceil 1/\varepsilon\rceil\right)$, the memory requirements for both the residual matrix $\hat{\bm{C}}$ and the queue of snapshots $\mathcal{S}$ are $O(d/\varepsilon)$, leading to a total memory cost of $O(d/\varepsilon)$.

\begin{algorithm}[h]
	\caption{ \dsfd: \textsc{Update}($\bm{a}_i$)}
    \label{alg:norm-seq-sw-fd-update}
	\KwIn{$\bm{a}_i$: the row vector arriving at timestamp $i$\\}

    \If{$i\equiv 1 \mod N$}{
        $\hat{\bm{C}} \leftarrow \hat{\bm{C}}^\prime$\\
        $\hat{\bm{C}}^\prime \leftarrow \bm{0}_{\ell\times d}$\\
        $\mathcal{S} \leftarrow \mathcal{S}$\\
        $\mathcal{S}^\prime \leftarrow \mathsf{queue}.\textsc{Initialize()}$\\
    }
    \While(\tcp*[f]{oldest snapshot expired}){$\mathcal{S}[0].t+N\le i$}{
        $\mathcal{S}.\textsc{popleft}()$\\
    }
    $\hat{\bm{C}}\leftarrow \mathsf{FD}_\ell(\hat{\bm{C}}, \bm{a}_i)$  \\
    \label{algln:norm-dumped}
    \While{$\lVert \hat{\bm{c}}_1 \rVert_2^2 \ge \theta$}{ 
        $\mathcal{S}$ append snapshot $(\bm{v}=\hat{\bm{c}}_1, s=\mathcal{S}[-1].t+1, t=i)$\\
        Remove first row $\hat{\bm{c}}_1$ from $\hat{\bm{C}}$ \tcp*[f]{$\hat{\bm{c}}_2^\prime$ becomes $\hat{\bm{c}}_1^\prime$} \\
    }
    $\hat{\bm{C}}^\prime \leftarrow \mathsf{FD}_\ell(\hat{\bm{C}}^\prime, \bm{a}_i)$  \\
    \While{$\lVert \hat{\bm{c}}^\prime_1 \rVert_2^2 \ge \theta$}{ 
        $\mathcal{S}^\prime$ append snapshot $(\bm{v}=\hat{\bm{c}}^\prime_1, s=\mathcal{S}^\prime[-1].t+1, t=i)$\\
        Remove first row $\hat{\bm{c}}^\prime_1$ from $\hat{\bm{C}}^\prime$ \tcp*[f]{$\hat{\bm{c}}_2^\prime$ becomes $\hat{\bm{c}}_1^\prime$} \\
    }
\end{algorithm}

\header{\bf Update algorithm.} Algorithm~\ref{alg:norm-seq-sw-fd-update} outlines the pseudocode for updating a \dsfd sketch, under the assumption that all arriving row vectors $\bm{a}_i$ are normalized, i.e., $\lVert \bm{a}_i\rVert_2^2 = 1$. The process starts by removing the oldest snapshots from the main queue until no snapshot in the queue is expired (lines 6-7). The input row vector then updates both the main and auxiliary FD sketches. If the squared norm of the top singular value multiplied by its corresponding right singular vector, $\lVert \hat{\bm{c}}_1 \rVert_2^2$, in any FD sketch surpasses the error threshold (evaluated in lines 9 and 13), this component is recorded alongside the current timestamp as a snapshot at the end of the respective queue (lines 10 and 14), before its removal from the FD sketch (lines 11 and 15).

Lines 1 to 5 describe the procedure of swapping the current main sketch and its queue with the auxiliary sketch and its queue, alongside initializing a new empty auxiliary sketch and queue of snapshots every $N$ update step. This methodology, named as \textit{restart every $N$ steps}, contributes to maintaining the algorithm's error bound which is detailed in the proof of Theorem~\ref{thm:norm-seq-sw-fd}.

The predominant time complexity for Algorithm~\ref{alg:norm-seq-sw-fd-update} arises from executing two FD updates, each necessitating an SVD on an $\ell \times d$ matrix. Hence, the total time complexity for each update step of \dsfd is $O(d\ell^2)$.

\header{\bf Optimized Update Algorithm. } We introduce optimization strategies to enhance the update algorithm for the primary \dsfd. The efficiency of each update step in Algorithm~\ref{alg:norm-seq-sw-fd-update} is dominated by the computation of $\textsf{FD}_{\ell}(\hat{\bm{C}}, \bm{a}_i)$, which traditionally requires $O(d\ell^2)$ time. Ghashami et al.~\cite{ghashami2016frequent} have introduced the \textsf{Fast-FD} technique, effectively reducing the update time complexity of FD from $O(d\ell^2)$ to an amortized $O(d\ell)$. This optimization is achieved by performing a merge operation between the matrix formed by the arrival of $\ell$ vectors and the FD sketch every time $\ell$ vectors accumulate, updating the FD sketch with the merged result.

Expanding upon this principle, we propose the \fastdsfd algorithm, which decreases the time complexity for a single update operation in \dsfd from $O(d\ell^2)$ to an amortized $O(d\ell+\ell^3)$, without incurring additional spatial complexity. Particularly for scenarios where $d = \Omega(\ell^2)$, the amortized complexity is effectively $O(d\ell)$.

Algorithm~\ref{alg:norm-seq-sw-fd-update2} details the optimized update operation for \fastdsfd, adopting the \textsf{Fast-FD} strategy by postponing the SVD operation until $\ell$ rows are prepared for merging. This deferred condition is met once every $\ell$ iterations, allocating the majority of the computational effort to the SVD (line 6) and the computation of the new matrix $\bm{K} = \hat{\bm{C}}\hat{\bm{C}}^\top$ (line 10). Each of these steps requires $O(d\ell^2)$ time, thereby resulting in an amortized time cost of $O(d\ell)$. This approach significantly enhances the update mechanism within the sliding window framework, ensuring the algorithm remains efficient and effective.

\begin{algorithm}[h]
	\caption{\textsf{\fastdsfd}: Optimized \textsc{Update} algorithm on Algorithm \ref{alg:norm-seq-sw-fd-update}.}
    \label{alg:norm-seq-sw-fd-update2} 
	\KwIn{$\bm{a}_i$: the row vector arriving at timestamp $i$\\}
    
    \While(\tcp*[f]{oldest snapshot expired}){$\mathcal{S}[0].t+N\le i$}{
        $\mathcal{S}.\textsc{popleft}()$\\
    }
    $\bm{D}\leftarrow \begin{bmatrix}
        \hat{\bm{C}}^\top& \bm{a}_i^\top
    \end{bmatrix}^\top$\\
    $\hat{\sigma_1} \leftarrow \sqrt{\hat{\sigma_1}^2+\lVert \bm{a}_i \rVert _2^2}$ \\
    
    \If{\#(rows of $\hat{D}$) $\ge 2\ell$}{
        $\hat{\bm{C}}\leftarrow \mathsf{FastFD}_\ell(\bm{D})$  \\
        \While{$\lVert \hat{\bm{c}}_1 \rVert_2^2 \ge \theta$}{
            $\mathcal{S}$ append snapshot $(\bm{v}=\hat{\bm{c}}_1, s=\mathcal{S}[-1].t+1, t=i)$\\
            Remove first row $\hat{\bm{c}}_1$ from $\hat{\bm{C}}$. \tcp*[f]{$\hat{\bm{c}}_2^\prime$ becomes $\hat{\bm{c}}_1^\prime$} \\
        }
        $\bm{K} = \hat{\bm{C}}\hat{\bm{C}}^\top$ \\
        $\hat{\sigma_1}\leftarrow \lVert \hat{\bm{c}}_1 \rVert_2$
    }
    \Else{
        $\bm{K} \leftarrow \begin{bmatrix}
            \bm{K} & \bm{a}_i \hat{\bm{C}}^\top\\
            \hat{\bm{C}} \bm{a}_i^\top & \bm{a}_i\bm{a}_i^\top\\
        \end{bmatrix}$\\
        \If{$\hat{\sigma_1}^2\ge \theta$}{
            $[\bm{U},\bm{\Sigma}^2,\bm{U}^\top]\leftarrow \mathsf{svd}(\bm{K})$\\
            \For{$j \in [1,\ell]$}{
                \If(\tcp*[f]{largest singular value}){$\sigma_j^2 \ge \theta$}{
                    $\bm{v}_j^\top \leftarrow \frac{1}{\sigma_j} \bm{u}_j^\top \bm{D} $ \\
                    $\mathcal{S}$ append snapshot $(\bm{v}=\bm{v}_j, s=\mathcal{S}[-1].t+1, t=i)$\\
                    $\bm{D} \leftarrow \bm{D} - \bm{Dv}_j\bm{v}_j^\top$\\
                    $\bm{K} \leftarrow \bm{K} - (\bm{D}\bm{v}_j)(\bm{D}\bm{v}_j)^\top$
                }
            }
            $\hat{\sigma_1}\leftarrow \sigma_1$
        }
        $\bm{\hat{C}}\leftarrow \bm{D}$
    }
\end{algorithm}

In the general scenario where the "if" condition remains untriggered, it is essential to determine whether the maximum singular value $\sigma_1$ of the FD sketch $\bm{D} = \begin{bmatrix} \hat{\bm{C}}^\top & \bm{a}_i^\top \end{bmatrix}^\top$, after incorporating the newly arrived row vector $\bm{a}_i$, exceeds the threshold $\theta$ (line 17). If this criterion is met, the corresponding right singular vector $\bm{v}_1$ along with the current timestamp $i$ is archived as a snapshot (line 19). This step also involves reducing the influence of this right singular vector $\bm{v}_1$ from both the FD sketch $\bm{D}$ and the covariance matrix $\bm{K}$, incurring a time cost of $O(d\ell + \ell^2)$ (lines 20, 21).

Rather than performing an SVD directly on the FD sketch $\bm{D}$ to compute $\bm{D} = \bm{U\Sigma V}^\top$ within $O(d\ell^2)$, an incremental update to $\bm{K} = \bm{D}\bm{D}^\top$ is executed rank-1-wise in $O(d\ell)$ time (line 13). Following this, conducting SVD on $\bm{K}$ to obtain $\bm{K} = \bm{U\Sigma}^2\bm{U}^\top$ can be achieved in $O(\ell^3)$ (line 15). Then multiply the maximum singular value's corresponding left singular vector $\bm{u}_1$ with $\bm{D}$ to extract the right singular vector $\bm{v}_1^\top = \frac{1}{\sigma_1}\bm{u}_1^\top \bm{D}$ is completed in $O(d\ell)$ time (line 18).

Summarizing the aforementioned time expenditures, the amortized time complexity for a single update step in Algorithm~\ref{alg:norm-seq-sw-fd-update2} is established as $O(d\ell + \ell^3)$. Lemma~\ref{lem:fast-dsfd} supports that the action taken in line 20—specifically, removing the $j$-th row of $\bm{\Sigma V}^\top$—parallels the procedure described in line 9.

\begin{lem}
    \label{lem:fast-dsfd}
    If $\bm{D}=\bm{U\Sigma V}^\top$ and $\bm{D}^\prime = \bm{D}-\bm{Dv}_j\bm{v}_j^\top$, where $\bm{v}_j$ is one of row vector of $\bm{V}^\top$. Then $\bm{D}=\bm{U\Sigma V}^\top(\bm{I}-\bm{v}_j\bm{v}_j^\top)$, which is same as remove the $j$-th row of $\bm{\Sigma V}^\top$.
\end{lem}

Given that the covariance matrix $\bm{K}$ might have $m$ singular values surpassing the threshold $\theta$, arranged as $\sigma_1^2 \ge \sigma_2^2 \ge \dots \ge \sigma_m^2 \ge \theta$ increasingly, the "if" condition on line 17 within the loop could be validated $m$ times, leading to a total operation count of $O(md\ell)$. For the normalized model, these operations can be averaged over $m\theta$ steps, yielding an amortized time complexity for the loop of $O(d\ell/\theta)$. With $\theta = N/\ell$ as specified in Problem~\ref{prob:norm-seq-sw-fd}, this amortization results in $O(d\ell^2/N)$. Assuming $N = \Omega(\ell)$, the amortized time complexity for a single update step in Algorithm~\ref{alg:norm-seq-sw-fd-update2} remains $O(d\ell + \ell^3)$, which is a reasonable assumption.

Setting $\ell = 1/\varepsilon$, the residual matrix $\hat{\bm{C}}$ of \fastdsfd can extend up to $2\ell \times d$, and the covariance matrix $\bm{K}$ can reach dimensions of $2\ell \times 2\ell$, leading to a space complexity of $O(\ell d + \ell^2)$. When $2\ell \le d$, the space complexity simplifies to $O(\ell d)$; for $2\ell > d$, maintaining the residual matrix $\hat{\bm{C}}$ up to $d \times d$ suffices, keeping the space complexity at $O(d^2 + \ell d) = O(\ell d)$. Thus, compared to \dsfd, \fastdsfd improves the update operation without incurring additional asymptotic space complexity.

Furthermore, potential optimizations could be realized. By maintaining an upper bound estimate $\hat{\sigma_1}$ on the maximum singular value $\sigma_1$ of the current FD sketch and updating this estimate to $\sqrt{\hat{\sigma_1}^2 + \lVert \bm{a}_i \rVert_2^2}$ with each new vector, an SVD is warranted only if this updated estimate exceeds the dump threshold $\theta$ (line 16), possibly avoiding the $O(\ell^3)$ SVD computation at this stage. Since only the largest singular value and its associated singular vector $\bm{u}_1$ of $\bm{K}$ are needed, iterative eigenvalue methods like Power Iteration could be used to reduce the time complexity of SVD further. Alternatively, randomized Krylov methods, such as Block Krylov, could offer a high-probability rank-1 approximation, requiring $O(\log(\min(1/\varepsilon, d)/\varepsilon'^{1/2}))$ matrix-vector multiplications, where $\varepsilon'$ denotes the relative error of the rank-1 approximation. This adjustment might render \fastdsfd a probabilistic algorithm with an amortized update time complexity of $O(d\ell)$, marking a substantial efficiency enhancement~\cite{bakshi2023krylov,musco2015randomized}.

\header{\bf Query algorithm.} The query operation for \dsfd and \fastdsfd is elegantly simple, as delineated in Algorithm~\ref{alg:norm-seq-sw-fd-query}. It involves merging the FD sketch $\hat{\bm{C}}$ with matrix $\bm{B}$ stacked by the vectors of non-expiring snapshots. This process ensures that the query algorithm efficiently utilizes the preserved historical data within the sliding window, alongside the current sketch, to provide accurate matrix approximations.

\begin{algorithm}[h]
	\caption{\dsfd and \fastdsfd: \textsc{Query}()}
    \label{alg:norm-seq-sw-fd-query}
    \Return $\mathsf{FD}_\ell (\bm{B},\hat{\bm{C}})$, where $\bm{B}$ is stacked by $s_i.\bm{v}$ for all $s_i\in \mathcal{S}$
\end{algorithm}

\subsection{Analysis}

We present the following theorem about the error guarantee, space usage, and update cost of \dsfd and \fastdsfd.

\begin{theorem}
    \label{thm:norm-seq-sw-fd}
    Assume that the data stream of vectors is $\bm{A} = \left[\bm{a}_1, \bm{a}_2,...,\bm{a}_T\right]$ and $\forall{1\le i\le T}, \lVert \bm{a}_i \rVert_2 = 1$.
    Given the length $N$ of the sliding window and the relative error $\varepsilon$, the \dsfd or \fastdsfd algorithm returns a sketch matrix $\bm{B}_W$. If we set $\theta=\varepsilon N$ and $\ell=\min\left(\lceil \frac{1}{\varepsilon}\rceil,d\right)$, then we have:
    \begin{equation}
        \label{eq:norm-seq-swfd-error}
        \textup{\textbf{cova-err}}(\bm{A}_W, \bm{B}_W)=\left\lVert \bm{A}_{T-N,T}^{\top}\bm{A}_{T-N,T} - \bm{B}^{\top}_W\bm{B}_W \right\rVert_2 \le 4\varepsilon N,
    \end{equation}
    where $\bm{A}_W=\bm{A}_{T-N,T} = 
    \left[\bm{a}_{T-N+1},\bm{a}_{T-N+2},\dots,\bm{a}_{T}\right]^\top$. The \dsfd or \fastdsfd algorithm uses $O\left(\frac{d}{\varepsilon} \right)$ space and process an update in $O\left(d\ell^2 \right)$ time for \dsfd, $O\left(d\ell+\ell^3 \right)$ time for \fastdsfd or $O\left(d\ell \right)$ time for probablistic \fastdsfd.
\end{theorem}

The proof of Theorem~\ref{thm:norm-seq-sw-fd} can be found in~\ref{proof:norm-seq-sw-fd}.

\section{Sequence-based model}
\label{sec:seq-dsfd}

In this section, we delve into extending \dsfd to \seqdsfd to accommodate the general case where row vectors are not normalized, that is, \(\lVert \bm{a}_i \rVert_2^2 \in [1, R]\), aiming to tackle Problem~\ref{prob:seq-sw-fd}. \seqdsfd achieves the optimal space complexity of \(O\left({d\over \varepsilon} \log R\right)\) for Problem~\ref{prob:seq-sw-fd}, aligning with the problem's lower bound space complexity as established in Theorem~\ref{thm:seq-swfd-lower-bound}.

\subsection{High-Level Ideas}

In the unnormalized scenario of sequence-based matrix sketching, the error bounds \(\varepsilon \lVert \bm{A}_W\rVert_F^2\) fluctuate based on the distribution of the most recently arrived vectors \(\bm{a}\in W\) over different periods. The minimum error bound is \(\varepsilon N\) assuming that \(\lVert \bm{a}\rVert_2^2=1\) for any \(\bm{a}\in W\), whereas the maximum error bound reaches \(\varepsilon NR\), assuming \(\lVert \bm{a}\rVert_2^2=R\) for any \(\bm{a}\in W\).

Drawing inspiration from the extension process applied in transitioning from the \textsc{BasicCounting} problem to the \textsc{Sum} problem~\cite{datar2002maintaining}, we regard the arrival of a row vector \(\bm{a}_i\) with \(\lVert \bm{a}_i\rVert_2^2=v_i\) as analogous to the simultaneous arrival of \(v_i\) normalized vectors \(\bm{a}_i/\sqrt{v_i}\). Following this approach, we apply the same update procedure as utilized in the normalized \dsfd. As a result, the outputs produced by the normalized \dsfd are confined within the window size of the most recent \(\lVert \bm{A}_W\rVert_F^2=\sum_{i=t_\text{now}-N+1}^{t_\text{now}}v_i\) normalized row vectors.

\begin{figure}
    \centering
    \includegraphics[width=\linewidth]{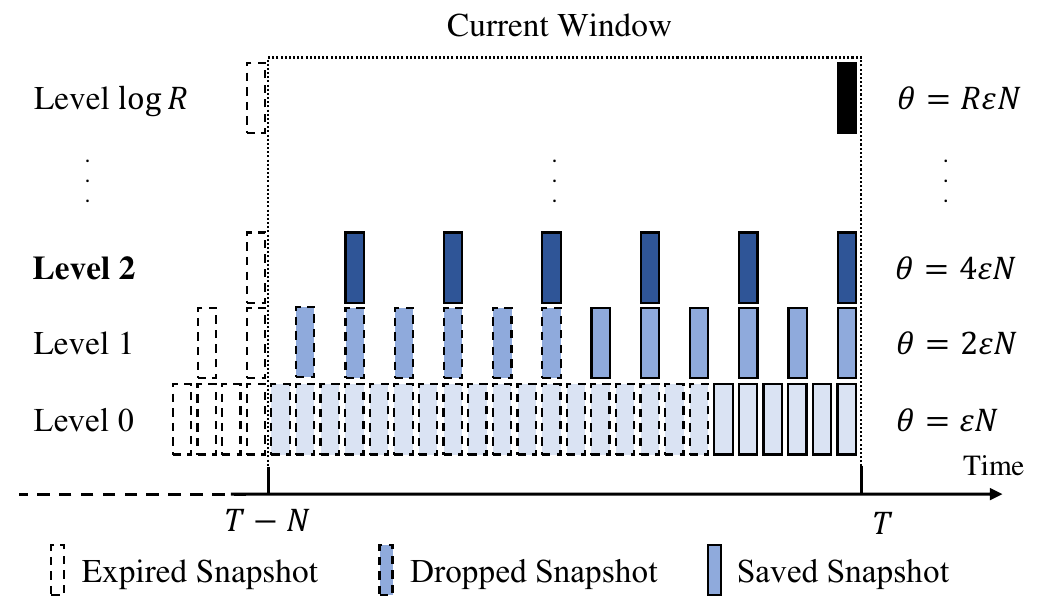}
    \caption{Sequence-based \dsfd. We maintain $L=\lceil\log R \rceil$ layers of \dsfd~ structures in parallel, each with different error bounds and dump thresholds $\theta=2^j\varepsilon N$ for the $j$-th level. In the visualization, depicted in dark blue, the norm of snapshots increases. For each level, we retain only the most recent $O\left({1\over \varepsilon}\right)$ snapshots saved in the queue and discard the older ones to limit the total memory usage to $O\left({1\over \varepsilon} \log R\right)$.}
    \label{fig:seq-dsfd}
\end{figure}

Given that the initial algorithm is limited to fixed-size sliding windows with a constant size of \(N\), extending normalized \dsfd to general \dsfd effectively broadens the algorithm's capability to maintain and provide answers for matrix covariance estimation over sliding windows of variable sizes. Inspired by the extended \(\lambda\)-snapshot scheme for the \textsc{FrequentItems} problem across variable-size sliding windows~\cite{lee2006simpler}, we introduce the \seqdsfd algorithm to address Problem~\ref{prob:seq-sw-fd}. Illustrated in Figure~\ref{fig:seq-dsfd}, our approach involves maintaining \(L=\lceil\log R \rceil\) layers of \dsfd structures concurrently, each configured with distinct error bounds and dump thresholds \(\theta=2^j\varepsilon N\) for the \(j\)-th level, where \(0\le j \le L\). Accordingly, the error bounds and dump thresholds extend from \(\varepsilon N, 2\varepsilon N,\dots, \varepsilon NR\) in an ascending sequence. The lower levels, characterized by smaller \(\theta\) values, perform dump operations more frequently.

When handling query requests for a sketch \(\bm{B}_W\) corresponding to the current window, we select the layer that meets the error bound criterion based on \(\varepsilon\lVert \bm{A}_W\rVert_F^2\) for the prevailing window. Notably, it is unnecessary to actively maintain \(\lVert \bm{A}_W\rVert_F^2\) or its approximation. Instead, we opt for the \dsfd layer that offers the minimal error bound while ensuring that the retained dumped snapshots span the time range of the current window size \(N\). This chosen layer then yields the appropriate sketch \(\bm{B}_W\).

\subsection{Algorithm descriptions. }

\begin{algorithm}[h]
	\caption{\seqdsfd: \textsc{Initialize}($d,\ell, N, R, \beta$)}
    \label{alg:seq-sw-fd-init} 
	\KwIn{
        $d$: Dimension of input vectors of \textsc{Update} \\
        $\ell=\min\left(\lceil\frac{1}{\varepsilon}\rceil, d\right)$: Number of rows in FD sketch \\
        $N$: Length of sliding window \\
        $R$: Upper bound of $\lVert \cdot \rVert_2^2$ for rows. For \textsc{Problem} \ref{prob:seq-sw-fd}, $R=\lVert\bm{A}\rVert_{2,\infty}^2$ \\
        $\beta$: Additional coefficient of error, default as 1.0 \\
    }
    $L\leftarrow \lceil \log_2 R \rceil$ \\
    $\mathcal{L}\leftarrow []$\\
    \For{$j\in [0,L]$}{
        $\mathcal{L}.\textsc{append}(\mathsf{Fast-DS-FD}(d,\ell, N, \theta=2^j\varepsilon N))$
    }
\end{algorithm}

\begin{algorithm}[h]
	\caption{\seqdsfd: \textsc{Update}($\bm{a}_i$)}
    \label{alg:seq-sw-fd-update} 
	\KwIn{$\bm{a}_i$: the row vector arriving at timestamp $i$\\}
 
    \For{$j \in [0,L]$}{
        \While{$len(\mathcal{L}[j].\mathcal{S}) > 2(1+\frac{4}{\beta})\frac{1}{\varepsilon}$ \textbf{or} $\mathcal{L}[j].\mathcal{S}[0].t+N\le i$}{
            $\mathcal{L}[j].\mathcal{S}$.\textsc{popleft}()\\
        }
        \If{$\lVert \bm{a}_i\rVert_2^2\ge 2^j \varepsilon N$}{
            $\mathcal{L}[j].\mathcal{S}$ appends $(\bm{v}=\bm{a}_i, s=\mathcal{L}[j].\mathcal{S}[-1].t+1, t=t_{\text{now}})$ \\
            $\mathcal{L}[j].\mathcal{S}^\prime$ appends  $(\bm{v}=\bm{a}_i, s=\mathcal{L}[j].\mathcal{S}^\prime[-1].t+1, t=t_{\text{now}})$ \\
        }\Else{
            $\mathcal{L}[j].\textsc{Update}(\bm{a}_i)$\\
        }
    }
\end{algorithm}

\header{\bf Data structures.} As outlined in Algorithm \ref{alg:seq-sw-fd-init}, we determine the number of levels to be \(L = \lceil \log_2 R\rceil\) (line 1). For each level \(i\), we initialize a \dsfd structure with a dump threshold \(\theta = 2^{i}\varepsilon N\) (lines 3-4). This setup specifies that the top row \(c_1\) of the FD sketch is dumped as a snapshot if \(\lVert c_1\rVert_2^2 \ge 2^{i}\varepsilon N\). Concurrently, to maintain memory efficiency, we cap the number of snapshots at each level to a maximum of $2(1+4/\beta)\frac{1}{\varepsilon}$, thereby constraining the memory requirement to \(O\left({d\over \varepsilon} \log R\right)\).

\header{\bf Update algorithm.} Algorithm \ref{alg:seq-sw-fd-update} details the procedure for updating a \seqdsfd sketch. Note that all incoming row vectors \(\bm{a}_i\) have norms \(\lVert \bm{a}_i\rVert_2^2\) within the range \([1,R]\). Initially, in line 3, we discard the oldest snapshot from the queue until the head element of the queue remains within the current sliding window and ensure the snapshot count does not surpass \(2(1+4/\beta)\frac{1}{\varepsilon}\). Subsequently, each level of \dsfd is updated with the input row vector \(\bm{a}_i\). If the norm of the input row vector \(\lVert \bm{a}_i\rVert_2^2\) is above the dump threshold \(\theta\) (line 4), we directly save it as a snapshot and add it to the queue of the corresponding level (lines 5 and 6), effectively achieving zero error for this row. If not, the \dsfd is updated with \(\bm{a}_i\) following the same procedures outlined in Algorithms~\ref{alg:norm-seq-sw-fd-update} or~\ref{alg:norm-seq-sw-fd-update2} (line 8).

For every incoming row \(\bm{a}_i\), Algorithm~\ref{alg:seq-sw-fd-update} processes the \fastdsfd sketch up to \(\log R\) times across all \(L=\log R\) layers, sum up to a total per-step update time of \(O\left(d\ell \log R \right)\).

In \seqdsfd, it remains essential to concurrently manage dual sets of sketches for each level and execute a \textit{restart every \(N\) step operation}, akin to the approach in \dsfd. Distinctively, for the primary \dsfd sketch at layer \(j\) within \seqdsfd, an alternating swap between the primary \dsfd sketch and the auxiliary \dsfd sketch occurs once the cumulative size of input vectors \(\sum \lVert \bm{a}_i \rVert_F^2\) processed by the primary \dsfd sketch surpasses \(2^{j+1} N\). This adaptation accounts for the perception of the arrival of a row vector \(\bm{a}_i\) with \(\lVert \bm{a}_i\rVert_2^2=v_i\) as equivalent to the simultaneous arrival of \(v_i/2^j\) rescaled vectors \(\sqrt{\frac{2^j}{v_i}}\bm{a}_i\). For clarity, this mechanism is not explicitly depicted in the pseudocode of Algorithms \ref{alg:seq-sw-fd-init} and \ref{alg:seq-sw-fd-update}.

\header{\bf Query algorithm.} Algorithm~\ref{alg:seq-sw-fd-query} details the procedure for generating a matrix sketch for the window \([t-N, t]\) utilizing a \seqdsfd sketch. Figure~\ref{fig:seq-dsfd} illustrates the possible state of the data structure at a specific moment.

Given that the queue of each layer saves \(O(\ell)\) snapshots, it is uncertain whether the snapshots preserved in the lower layers encompass the full window, as shown in levels 0 and 1 of Figure~\ref{fig:seq-dsfd}. Consequently, it is imperative to identify the lowest layer (offering minimal error) that contains snapshots spanning a window length of \(N\) (as exemplified by level 2 in Figure~\ref{fig:seq-dsfd}). A straightforward linear search could be conducted by verifying if the timestamp of the oldest snapshot in each layer's queue falls within the window interval \([t-N, t]\), which would entail a time complexity of \(O(\log R)\). However, considering that the timestamp of the leading snapshot in each layer's queue increases monotonically, a binary search method could be employed, effectively reducing the time complexity to \(O(\log \log R)\).

\begin{algorithm}[h]
	\caption{\seqdsfd: \textsc{Query}()}
    \label{alg:seq-sw-fd-query}
    Find the $\min_{j} 1\le \mathcal{L}[j].\mathcal{S}[0].s \le t_{\text{now}}-N+1$\\
    \Return $\mathsf{FD}_\ell (\bm{B}, \mathcal{L}[j].\hat{\bm{C}})$, where $\bm{B}$ is stacked by $s_i.\bm{v}$ for all $s_i\in \mathcal{L}[j].\mathcal{S}$
\end{algorithm}

\subsection{Analysis}

We present the following theorem about the error guarantee, space usage, and update cost of \seqdsfd.

\begin{theorem}
    \label{thm:seq-sw-fd}
    Assume that the data stream of vector is $\bm{A} = \left[\bm{a}_1, \bm{a}_2,...,\bm{a}_T\right]$ and $\forall{1\le i\le T}, \lVert \bm{a}_i \rVert_2^2 \in [1,R]$.
    Given the length $N$ of the sliding window and the relative error $\varepsilon$, the \seqdsfd~ algorithm returns a sketch matrix $\bm{B}_W$.
    If we set $\ell=\min\left(\lceil \frac{1}{\varepsilon}\rceil, d\right)$ and $\beta>0$, then we have:
    \begin{equation}
        \label{eq:error-of-seq-swfd}
        \textup{\textbf{cova-err}}(\bm{A}_W, \bm{B}_W)=\left\lVert \bm{A}_W^{\top}\bm{A}_W - \bm{B}_W^{\top}\bm{B}_W \right\rVert_2 \le \beta \varepsilon \lVert \bm{A}_W \rVert_F^2 ,
    \end{equation}
    where $\bm{A}_W=\bm{A}_{T-N,T} = 
    [\bm{a}_{T-N+1},\bm{a}_{T-N+2},\dots,\bm{a}_{T}]^\top$. Suppose the $\beta$ is constant and $u$ as the update time of each level; the \seqdsfd algorithm uses $O\left(\frac{d}{\varepsilon}\log R \right)$ space and processes an update in $O\left(u \log R \right)$ time.
\end{theorem}

The proof of Theorem~\ref{thm:seq-sw-fd} can be found in~\ref{proof:thm:seq-sw-fd}.

\section{Time-based model}
\label{sec:time-dsfd}

In our previous discussions, we focused on \textit{sequence-based data streams}, characterized by the regular arrival of data items, where the arrival time increments by one unit for each arrival. However, in many real-world applications of sliding window data streams, attention often shifts to sliding windows defined in real-time terms---such as maintaining a sketch of data items received over the past hour or day. This scenario is termed \textit{time-based data streams}.

Time-based data streams differ from sequence-based streams primarily in two aspects: (1) The presence of \textit{idle} periods, which denote intervals without any arriving items, or equivalently, when the incoming vector $\bm{a}_t$ is the zero vector. This characteristic introduces potential sparsity within the current time window, i.e., $\lVert \bm{A}_W\rVert_F^2 < N$ (assuming $\bm{a}_t\in \{0\} \cup [1,R]$). (2) The occurrence of \textit{bursty} phenomena, characterized by abrupt increases in the rates of data item arrivals.

Adapting \seqdsfd to accommodate the time-based model is straightforward. Given that the error bound may be less than $\varepsilon N$ when $\lVert \bm{A}_W\rVert_F^2 < N$, we adjust the number of parallel \dsfd layers to $L=\lceil \log_2(\varepsilon NR)\rceil$, as opposed to $L=\lceil \log_2(R)\rceil$ employed in \seqdsfd. The dump thresholds are set to $\theta=2^i$ ($1, 2, 4, \dots, \varepsilon NR$) for all layers $0 \le i \le L$, thereby determining the memory cost as $O\left({d\over \varepsilon} \log(\varepsilon NR)\right)$. The procedures for updates and queries remain consistent with those described in Algorithm~\ref{alg:seq-sw-fd-update} and Algorithm~\ref{alg:seq-sw-fd-query}.

Here, $NR$ may also represent the maximal potential value that $\lVert \bm{A}_W\rVert_F^2$ could achieve within a time-based sliding window. In instances where the arriving vectors are normalized, i.e., $R=1$ and $\lVert \bm{a}_t \rVert_2^2 \in \{0,1\}$, the sketch's number of layers is modified to $L=\lceil \log(\varepsilon N)\rceil$, with the total size being $O\left({d\over \varepsilon} \log(\varepsilon N)\right)$.

\begin{corollary} 
    \label{coro:time-sw-fd}
    Assume that the data stream of vectors is $\bm{A} = \left[\bm{a}_1, \bm{a}_2,...,\bm{a}_T\right]$, and $\lVert \bm{a}_i \rVert_2^2 \in \{0\}\cup [1,R]$ for $\forall{1\le i\le T}$. Given the length $N$ of the sliding window and the relative error $\varepsilon$, the Time-\dsfd algorithm returns a sketch matrix $\bm{B}_W$.
    If we set $\ell=\min\left(\lceil \frac{1}{\varepsilon}\rceil, d\right)$ and $\beta>0$, then we have:
    \begin{equation}
        \label{eq:error-of-time-swfd}
        \textup{\textbf{cova-error}}(\bm{A}_W, \bm{B}_W)=\left\lVert \bm{A}_W^{\top}\bm{A}_W - \bm{B}_W^{\top}\bm{B}_W \right\rVert_2 \le \beta \varepsilon \lVert \bm{A}_W \rVert_F^2,
    \end{equation}
    where $\bm{A}_W=\bm{A}_{T-N,T} = 
    [\bm{a}_{T-N+1},\bm{a}_{T-N+2},\dots,\bm{a}_{T}]^\top$. 
    Suppose the $\beta$ is constant and $u$ as the update time of each level; the Time-\dsfd algorithm uses $O\left(\frac{d}{\varepsilon}\log \varepsilon NR \right)$ space and processes an update in $O\left(u \log \varepsilon NR \right)$ time.
\end{corollary}
\begin{figure*}[t]
    \includegraphics[width=\textwidth]{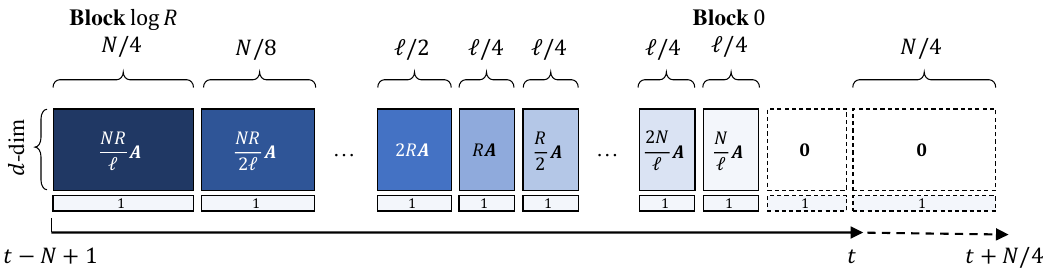}
    \caption{A constructive hard instance to establish a space lower bound for the sequence-based model. We initiate the sliding window's state by partitioning it into $\log R + 1$ blocks, each exponentially decreasing in size, and proceed to append one-hot vectors to the window over time. As each block expires, the algorithm is required to expend $\Omega\left( d \ell \right)$ bits to accurately estimate the expired block, according to Lemma~\ref{lem:fd-lower-bound}. Consequently, by considering the number of blocks $\log R + 1$, we derive the lower bound $\Omega\left(d \ell \log R\right)$. The rigorous proof is provided in the text for Theorem~\ref{thm:seq-swfd-lower-bound}. }
    \label{fig:seq-based-lower-bound}
\end{figure*}

\section{Space Lower bound}

In this section, we delve into the lower bounds of space requirements for any deterministic algorithm designed to address the problem of matrix sketching over sliding windows. Our analysis aims to demonstrate that the space complexities of our proposed algorithms align with these lower bounds. This alignment confirms the optimality of our algorithms in terms of memory requirements, showcasing their efficiency in handling the constraints imposed by sliding window contexts.

\begin{lem}
    \label{lem:fd-lower-bound}
    Let $\bm{B}$ be a $\ell \times d$ matrix approximating a $n \times d$ matrix $\bm{A}$ such that $\lVert \bm{A}^\top \bm{A} - \bm{B}^\top \bm{B}\rVert_2\le \lVert \bm{A} - \bm{A}_k\rVert_F^2/(\ell-k)$. For any algorithm with input as an $n\times d$ matrix $A$, the space complexity of representing $\bm{B}$ is $\Omega(d\ell)$ bits of space.
\end{lem}

Lemma~\ref{lem:fd-lower-bound}, as established by Ghashami et al.~\cite{ghashami2016frequent}, lays the groundwork for understanding the space efficiency of matrix sketching algorithms. Building upon this lemma, we aim to prove theorems concerning the space lower bounds for both sequence-based and time-based models of matrix sketching over sliding windows. These theorems ensure that our algorithms not only maintain precise sketching capabilities under the constraints of sliding window contexts but also adhere to the minimal possible space complexity.

\begin{theorem}[Seq-based Lower Bound]
    \label{thm:seq-swfd-lower-bound}
    A deterministic algorithm that returns $\bm{B}_W$ such that

\begin{equation*}
    \lVert \bm{A}_W^\top\bm{A}_W-\bm{B}_W^\top\bm{B}_W\rVert_2 \le \frac{\varepsilon}{3} \lVert \bm{A}_W\rVert_F^2,
\end{equation*}
    where $\varepsilon=1/\ell$, $\bm{A}_W\in \mathbb{R}^{N\times (d+1)}$, $N\ge \frac{1}{2\varepsilon}\log \frac{R}{\varepsilon}$ and $1\le \lVert \bm{a} \rVert_2^2 \le R+1$ for all $\bm{a} \in \bm{A}_W$ must use $\Omega\left(\frac{d}{\varepsilon}\log R \right)$ bits space.
\end{theorem}

\begin{proof}
    
We partition a window of size $N$ consisting of $(d+1)$-dimensional vectors into $\log R +2$ blocks, as illustrated in Figure~\ref{fig:seq-based-lower-bound}. The leftmost $\log R+1$ blocks are labeled as $\log R, \dots, 1,0$ from left to right, as depicted in Figure~\ref{fig:seq-based-lower-bound}. The construction of these blocks is as follows: (1) Choose $\log R + 1$ matrices of size $\frac{\ell}{4} \times d$ from a set of matrices $\mathcal{A}$, where $\mathcal{A}$ ensures that $\bm{A}_i^\top\bm{A}_i$ is an $\ell/4$ dimensional projection matrix and $\lVert\bm{A}_i^\top \bm{A}_i - \bm{A}_j^\top\bm{A}_j\rVert > 1/2$ for all $\bm{A}_i,\bm{A}_j \in \mathcal{A}$. Ghashami et al.~\cite{ghashami2016frequent} have demonstrated the existence of such a set $\mathcal{A}$ with cardinality $\Omega(2^{d\ell})$, making the total number of distinct arrangements $L= {\Omega(2^{d\ell})\choose{\log R + 1}}$. Consequently, $\log L = \Omega(d\ell \log R)$. (2) For block $i$, multiply the chosen $\bm{A}_i \in \mathbb{R}^{\frac{\ell}{4} \times d}$ by a scalar of $\sqrt{\frac{2^i N}{\ell}}$, making the square of the Frobenius norm of block $i$, $\lVert\bm{A}_i\rVert_F^2$, equal to $2^i N/4$. (3) For block $i$ where $i > \log \frac{\ell R}{N}$, increase the number of rows from $\ell/4$ to $\frac{\ell}{4} \cdot 2^{i-\log\frac{\ell R}{N}}$ to ensure that $1\le \lVert \bm{a} \rVert_2^2 \le R$. The total number of rows is bounded by $N$, thus $\frac{N}{2}+\frac{\ell}{4}\log\frac{\ell R}{2N}\le N$, which implies $N \ge \frac{\ell}{2}\log \ell R$. (4) Set all the $(d+1)$-dimensional vectors in the window to be $1$.

We assume the algorithm is presented with one of these $L$ arrangements of length $N$, followed by a sequence of all one-hot vectors of length $N$ with only the $(d+1)$-dimension set as 1. We denote $\bm{A}_W^i$ as the matrix over the sliding window of length $N$ at the moment when $i+1, i+2,\dots, \log R$ blocks have expired. Suppose our sliding window algorithm provides estimations $\bm{B}_W^i$ and $\bm{B}_W^{i-1}$ of $\bm{A}_W^i$ and $\bm{A}_W^{i-1}$, respectively, with a relative error of $\frac{1}{3\ell}$. This implies

\begin{gather*}
    \lVert {\bm{A}_W^i}^\top\bm{A}_W^i -{\bm{B}_W^i}^\top {\bm{B}_W^i}\rVert_2\le \frac{1}{3\ell} \lVert \bm{A}_W^i\rVert_F^2 = \frac{1}{3\ell}\left( \frac{N}{4} \cdot 2^{i+1} + \frac{3N}{4}\right),\\
    \lVert {\bm{A}_W^{i-1}}^\top\bm{A}_W^{i-1} -{\bm{B}_W^{i-1}}^\top {\bm{B}_W^{i-1}}\rVert_2\le \frac{1}{3\ell} \lVert \bm{A}_W^{i-1}\rVert_F^2 = \frac{1}{3\ell}\left( \frac{N}{4} \cdot 2^{i} + \frac{3N}{4} \right).
\end{gather*}

Then we can answer the block $i$ with $\bm{B}_i^\top \bm{B}_i = {\bm{B}_W^i}^\top {\bm{B}_W^i} - {\bm{B}_W^{i-1}}^\top {\bm{B}_W^{i-1}}$ as below,

\begin{equation*}
    \label{eq:leftest-block}
    \begin{split}
        &\lVert \bm{A}_i^\top \bm{A}_i -\bm{B}_i^\top \bm{B}_i\rVert_2\\
        =&\lVert ({\bm{A}_W^i}^\top{\bm{A}_W^i} -{\bm{A}_W^{i-1}}^\top{\bm{A}_W^{i-1}}) -({\bm{B}_W^i}^\top {\bm{B}_W^i}  -{\bm{B}_W^{i-1}}^\top {\bm{B}_W^{i-1}})\rVert_2\\
        \le & \lVert {\bm{A}_W^i}^\top{\bm{A}_W^i} -{\bm{B}_W^i}^\top {\bm{B}_W^i}\rVert_2 +\lVert {\bm{A}_W^{i-1}}^\top{\bm{A}_W^{i-1}}  -{\bm{B}_W^{i-1}}^\top {\bm{B}_W^{i-1}}\rVert_2\\
        \le & \frac{1}{3\ell}\left(\frac{3}{4}N\cdot 2^i + \frac{3}{2}N \right) \le \frac{1}{\ell} \lVert \bm{A}_i\rVert_F^2.
    \end{split}
\end{equation*}

The algorithm is capable of estimating the blocks for levels $0 \le i \le \log_2 R$. According to Lemma~\ref{lem:fd-lower-bound}, each estimation of the blocks' levels by the algorithm necessitates $\Omega(d\ell)$ bits of space. Thus, we derive that a fundamental lower bound for the space complexity of any deterministic algorithm addressing the problem of matrix sketching over sliding windows is $\Omega\left(\frac{d}{\varepsilon} \log R \right)$.
\end{proof}

We also derive the space lower bound in the time-based model:

\begin{theorem}[Time-based Lower Bound]
    \label{thm:time-swfd-lower-bound}
    Any deterministic algorithm that returns $\bm{B}_W$ such that

\begin{displaymath}
    \lVert \bm{A}_W^\top\bm{A}_W-\bm{B}_W^\top\bm{B}_W\rVert_2 \le \frac{\varepsilon}{3} \lVert \bm{A}_W\rVert_F^2,
\end{displaymath}
    where $\bm{A}_W\in \mathbb{R}^{N\times d}$, $N\ge \frac{1}{2\varepsilon}\log \frac{R}{2}$ and $\lVert \bm{a} \rVert_2^2\in {0}\cup [1,R]$ for all $\bm{a} \in \bm{A}_W$, one of the space lower bound is $\Omega\left(\frac{d}{\varepsilon}\log \varepsilon NR \right)$.
\end{theorem}

The proof of Theorem~\ref{thm:time-swfd-lower-bound} is similar to that of Theorem~\ref{thm:seq-swfd-lower-bound} and can be found in~\ref{proof:time-swfd-lower-bound}.

Gathering Theorems \ref{thm:norm-seq-sw-fd}, \ref{thm:seq-sw-fd}, \ref{thm:seq-swfd-lower-bound}, and \ref{thm:time-swfd-lower-bound}, Corollary \ref{coro:time-sw-fd}, Lemma \ref{lem:fd-lower-bound}, we synthesize the space complexity results of the optimal algorithms we've proposed against the proven space lower bounds for any deterministic algorithm under four distinct models in the last two rows of Table~\ref{tab:alg}. 
\section{Experiment}

\subsection{Experiment Setup}

\header{\bf Datasets. } We conduct our experiments on both sequence-based and time-based models, utilizing a combination of synthetic and real-world datasets. For the sequence-based model, our experiments encompass one synthetic dataset and two real-world datasets. The characteristics and sources of these datasets are detailed below and summarized in Table~\ref{tab:seq-dataset}:

\begin{itemize}[leftmargin= * ]
    \item \textbf{SYNTHETIC:} This dataset is a Random Noisy matrix commonly used to evaluate matrix sketching algorithms, generated by the formula $\bm{A} = \bm{SDU} + \bm{N}/\zeta$. Here, $\bm{S}$ is a $n \times d$ matrix of signal coefficients, with each entry drawn from a standard normal distribution. $\bm{D}$ is a diagonal matrix with $\bm{D}_{i, i} = 1-(i-1)/d$. $\bm{U}$ represents the signal row space, satisfying $\bm{UU}^\top = \bm{I}_d$. The matrix $\bm{N}$ adds Gaussian noise, with $\bm{N}_{i,j}$ drawn from $\mathcal{N}(0, 1)$. We set $\zeta = 10$ to ensure the signal $\bm{SDU}$ is recoverable. The window size is set to $N=100,000$ for the SYNTHETIC dataset.
    
    \item \textbf{BIBD:}\footnote{\href{https://www.cise.ufl.edu/research/sparse/matrices/JGD_BIBD/bibd_22_8.html}{University of Florida Sparse Matrix Collection}} This dataset is the incidence matrix of a Balanced Incomplete Block Design by Mark Giesbrecht from the University of Waterloo. It consists of 231 columns, 319,770 rows, and 8,953,560 non-zero entries, with each entry being an integer (0 or 1) indicating the presence or absence of an edge. The window size is set to $N=10,000$ for the BIBD dataset.
    
    \item \textbf{PAMAP2 Physical Activity Monitoring:}\footnote{\href{https://archive.ics.uci.edu/dataset/231/pamap2+physical+activity+monitoring}{UCI Machine Learning Repository}} This dataset contains data from 18 different physical activities performed by 9 subjects wearing inertial measurement units and a heart rate monitor. For our experiments, we use data from subject 3, which includes 252,832 rows and 52 columns (timestamps and activity IDs removed, all missing entries set as 1). The window size is set to $N=10,000$ for the PAMAP2 dataset.
\end{itemize}

\begin{table}[h]
\caption{Datasets for the sequence-based window.}
\label{tab:seq-dataset}
\begin{tabular}{|l|r|r|r|r|}
\hline
Data Set  & Total Rows $n$ & $d$   & $N$     & Ratio $R$ \\ \hline
SYNTHETIC & $500,000$ & $300$ & $100,000$ & $14.75$   \\ \hline
BIBD      & $319,770$ & $231$ & $10,000$ & $1$       \\ \hline
PAMAP2    & $252,832$ & $52$  & $10,000$ & $1,403$        \\ \hline
\end{tabular}
\end{table}

We also evaluate the algorithms over the time-based model on two real-world datasets:

\begin{itemize}[leftmargin= * ]
    \item \textbf{RAIL}\footnote{\href{https://www.cise.ufl.edu/research/sparse/matrices/Mittelmann/rail2586.html}{University of Florida Sparse Matrix Collection}} dataset is the crew scheduling matrix for the Italian railways, where the entry at row $i$ and column $j$ denotes the integer cost for assigning crew $i$ to cover trip $j$. For our experiments, we selected a $200,000\times 500$ submatrix. Synthetic timestamps for RAIL were generated following the \textit{Poisson Arrival Process} with $\lambda=0.5$, setting the window size to $50,000$, which results in approximately $100,000$ rows on average per window.
    
    \item \textbf{YearPredictionMSD (YEAR)}\footnote{\href{https://archive.ics.uci.edu/dataset/203/yearpredictionmsd}{UCI Machine Learning Repository}} is a subset from the “Million Songs Dataset” \cite{bertin2011million} that includes the prediction of the release year of songs based on their audio features. It comprises over $500,000$ rows and $d = 90$ columns. For our analysis, a subset with $N = 200,000$ rows was utilized. This matrix exhibits a high rank. Synthetic timestamps for YEAR were similarly generated following the \textit{Poisson Arrival Process} with $\lambda=0.5$.
\end{itemize}

\begin{table}[h]
\caption{Datasets for the time-based window.}
\label{tab:time-dataset}
\begin{tabular}{|l|r|r|r|r|r|}
\hline
Data Set  & Total Rows $n$ & $d$   & $\Delta$  & $N_W$ & Ratio $R$ \\ \hline
RAIL     & $200,000$ & $500$  & $50,000$ &  $\approx 100,000$ & $12$ \\ \hline
YEAR     & $200,000$ & $90$  & $50,000$ &  $\approx 100,000$ & $1,321$ \\ \hline
\end{tabular}
\end{table}

\header{\bf Algorithms and parameters. }We compare \dsfd algorithm with three leading baseline competitors: Sampling algorithms, \lmfd, and \difd~\cite{wei2016matrix}. The evaluations of Sampling algorithms included \textsf{SWR} (row sampling with replacement) and \textsf{SWOR} (row sampling without replacement), where, for both strategies, sampling $\ell=O\left(d/\varepsilon^2 \right)$ rows is required to attain an $\varepsilon$-covariance error.

The \lmfd algorithm is tested for sequence-based and time-based window models. For \lmfd, the space and error metrics were adjusted by a parameter $b = 1/\varepsilon$, the number of blocks per level, and a parameter $\ell=\min\left(1/\varepsilon, d\right)$, the size of the approximation matrix $\bm{B}$ (i.e., the number of rows in $\bm{B}$), to achieve an $8\varepsilon$ relative covariance error. The evaluation of \difd is only conducted to the sequence-based sliding window model, with the space and error parameters determined by $L = \log (R/\varepsilon)$, the maximal number of levels within the dyadic interval framework.

Furthermore, both Seq-\dsfd and Time-\dsfd are tested across sequence-based and time-based windows. Similarly to \difd, it is necessary to estimate the maximal value $R$ of the row vector norm to set the number of layers as $L = \log R$ (for sequence-based) or $L=\log (\varepsilon N R)$ (for time-based). The parameters $\beta$ and the size of FD sketches $\ell$ are adjustable, balancing space complexity and the covariance relative error boundary.
    
\header{\bf Metrics.} In our experimental study, we adjust the parameters for each algorithm to illustrate the trade-offs between the \textit{maximum sketch size} and the observe \textit{maximum} and \textit{average error}.

\begin{itemize}[leftmargin= * ]
    \item \textbf{Sketch Size:} This metric denotes the space the matrix sketching algorithm occupies within the current window at a specific time. Considering that the primary part of the space cost comes from the row vectors of dimension $d$, we use the maximum number of rows to describe the space overhead for matrix sketching algorithms across different datasets.
    \item \textbf{Maximum and Average Relative Errors:} These metrics are employed to assess the quality of matrix estimates for various sketch algorithms. The relative error is defined as $\lVert \bm{A}_W^\top \bm{A}_W - \bm{B}_W^\top \bm{B}_W\rVert_2/\lVert \bm{A}_W\rVert_F^2$, where $\bm{A}_W$ represents the accurate matrix within the current sliding window, and $\bm{B}_W$ is the estimated matrix produced by the sketching algorithm.
\end{itemize}

\header{\bf Hardware.} For probabilistic algorithms such as random sampling, we employ the same random seed to guarantee the reproducibility of our experiments. All algorithms are implemented in Python 3.12.0. Experiments are conducted on a single idle core of an Intel\textregistered~ Xeon\textregistered~ CPU E7-4809 v4, clocked at 2.10 GHz.

\subsection{Experiment Results}

\begin{figure}[h]
    \centering
    \begin{subfigure}[b]{0.48\linewidth}
        \centering
        \includegraphics[width=\textwidth]{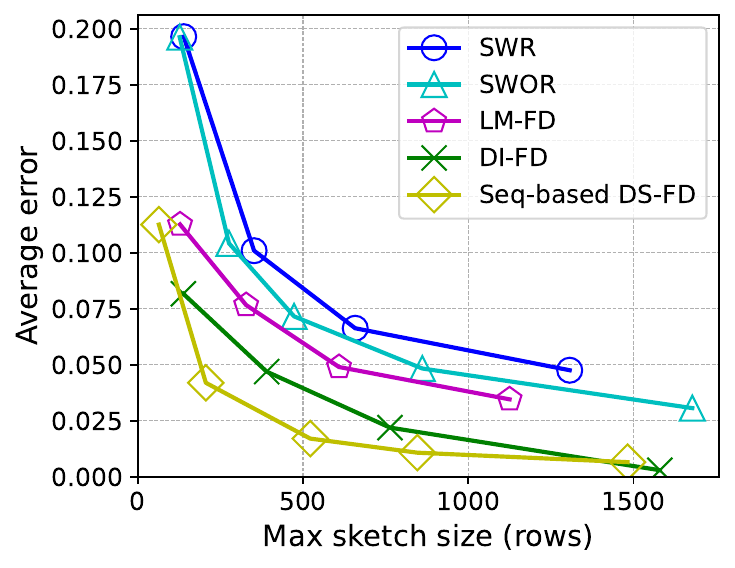}
        \caption{Average err vs. sketch size}
        \label{fig:synthetic-avg-err}
    \end{subfigure}
    \hfill
    \begin{subfigure}[b]{0.48\linewidth}
        \centering
        \includegraphics[width=\textwidth]{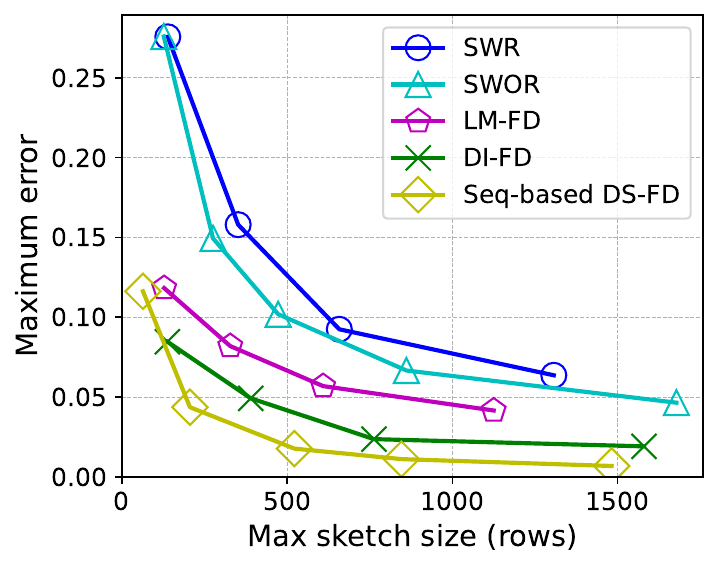}
        \caption{Maximum err vs. sketch size}
        \label{fig:synthetic-max-err}
    \end{subfigure}
    \caption{Error vs. sketch size on SYNTHETIC dataset.}
    \label{fig:exp-synthetic}
\end{figure}

\begin{figure}[h]
    \centering
    \begin{subfigure}[b]{0.48\linewidth}
        \centering
        \includegraphics[width=\textwidth]{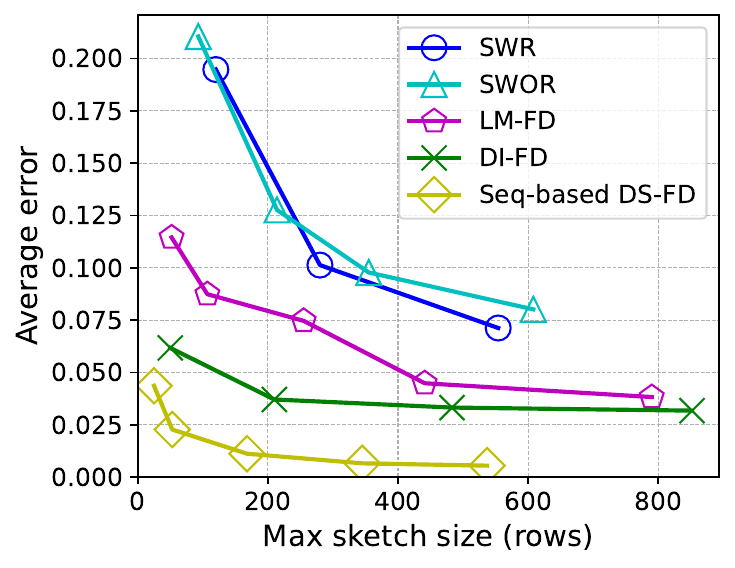}
        \caption{Average err vs. sketch size}
        \label{fig:bibd-avg-err}
    \end{subfigure}
    \hfill
    \begin{subfigure}[b]{0.48\linewidth}
        \centering
        \includegraphics[width=\textwidth]{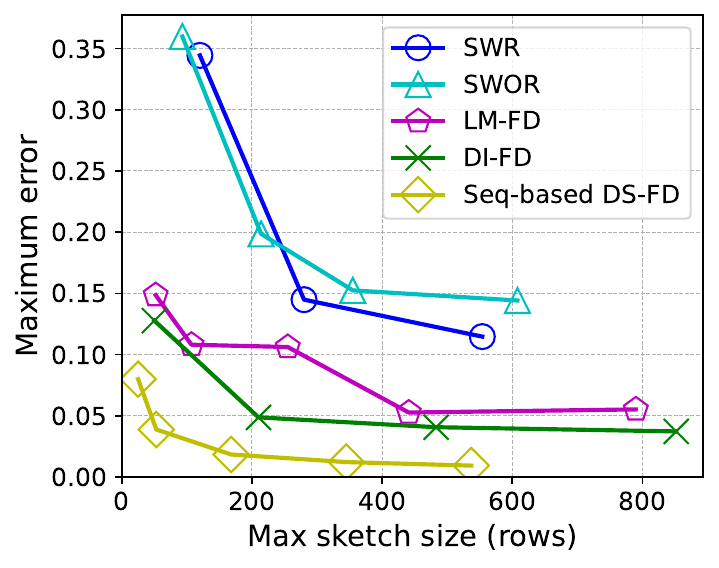}
        \caption{Maximum err vs. sketch size}
        \label{fig:bibd-max-err}
    \end{subfigure}
    \caption{Error vs. sketch size on BIBD dataset.}
    \label{fig:exp-bibd}
\end{figure}

\begin{figure}[h]
    \centering
    \begin{subfigure}[b]{0.48\linewidth}
        \centering
        \includegraphics[width=\textwidth]{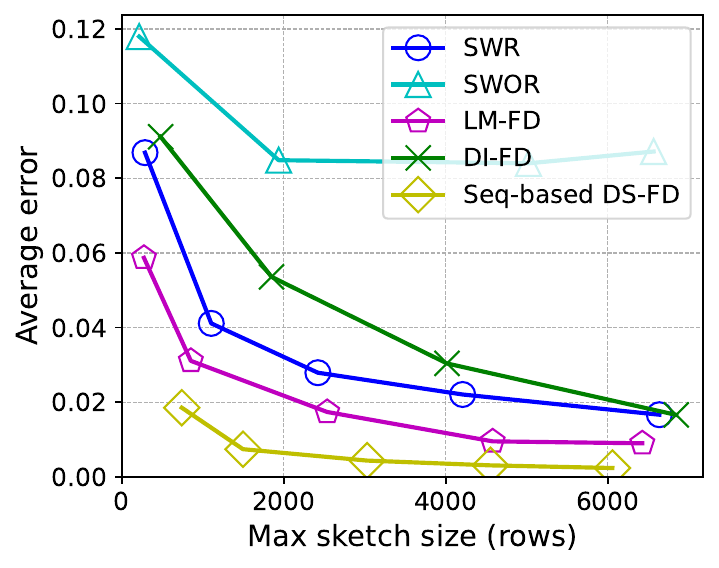}
        \caption{Average err vs. sketch size}
        \label{fig:pamap-avg-err}
    \end{subfigure}
    \hfill
    \begin{subfigure}[b]{0.48\linewidth}
        \centering
        \includegraphics[width=\textwidth]{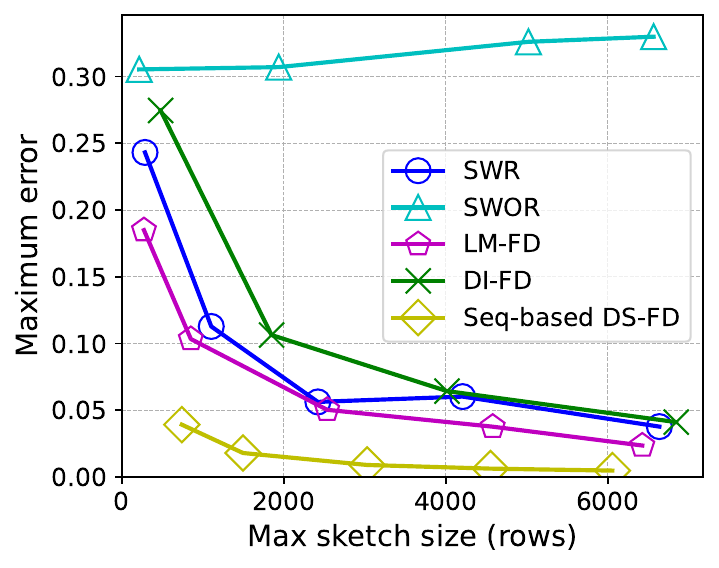}
        \caption{Maximum err vs. sketch size}
        \label{fig:pamap-max-err}
    \end{subfigure}
    \caption{Error vs. sketch size on PAMAP2 dataset.}
    \label{fig:exp-pamap}
\end{figure}

\header{\bf Errors vs. memory cost. } We begin our evaluation by comparing the empirical relative covariance error and memory cost across all algorithms. For each method, we adjusted a series of parameters to modulate the theoretical upper bound of the covariance error. We report the maximum sketch sizes, along with the average and maximum empirical relative covariance error across all queries. The evaluation was conducted on the three datasets for the sequence-based sliding window model. Figures~\ref{fig:exp-synthetic},~\ref{fig:exp-bibd}, and~\ref{fig:exp-pamap} illustrate the trade-offs between maximum sketch size and average error, and between maximum sketch size and maximum error, respectively. From these observations, we infer the following points:

(1) In the sequence-based scenario, the error-space trade-off of \lmfd, \difd, and \dsfd outperforms that of the sampling algorithms \textsf{SWR} and \textsf{SWOR}. Notably, \difd exhibits a performance decline in skewed data streaming (PAMAP2), as depicted in Figures~\ref{fig:pamap-avg-err} and~\ref{fig:pamap-max-err}, aligning with observations in~\cite{wei2016matrix}.
    
(2) The trade-off between error and space for \dsfd is consistently superior to other competitors across both synthetic and real-world datasets in the sequence-based model. \dsfd achieves the same covariance error with less memory overhead than other methods. Furthermore, no empirical error was observed to exceed the theoretical bound, i.e., $\lVert \bm{A}_W^\top\bm{A}_W-\bm{B}_W^\top\bm{B}_W\rVert_2 > \varepsilon \lVert \bm{A}_W \rVert_F^2$, affirming our theoretical analysis and underscoring the efficiency and correctness of our algorithm.
    
(3) The trade-off advantage of \dsfd becomes more pronounced as the ratio $R$ (the ratio between the maximum and minimum squared norms in the dataset) decreases. For example, in Figures~\ref{fig:bibd-avg-err} and~\ref{fig:bibd-max-err}, when row vectors are fully normalized in the BIBD dataset, \dsfd consumes significantly less memory to achieve a certain level of covariance error compared to other algorithms. Additionally, the average and maximum relative errors nearly approach 0 when the maximum sketch size exceeds 200 rows.

\begin{figure}[h]
    \centering
    \begin{subfigure}[b]{0.45\linewidth}
        \centering
        \includegraphics[width=\textwidth]{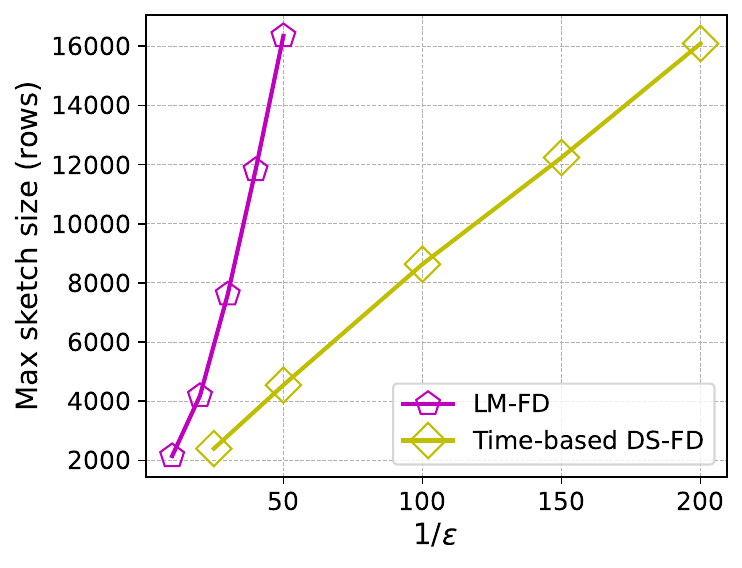}
        \caption{RAIL}
        \label{fig:rail-l-memory}
    \end{subfigure}
    \begin{subfigure}[b]{0.45\linewidth}
        \centering
        \includegraphics[width=\textwidth]{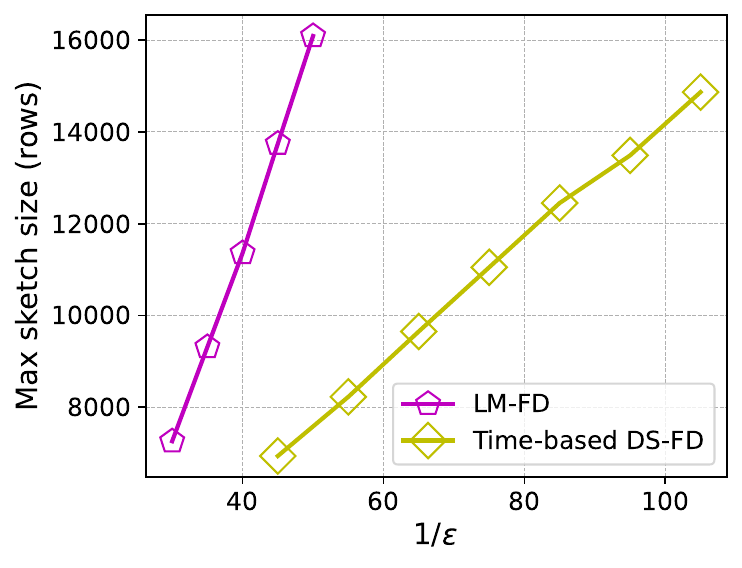}
        \caption{YEAR}
        \label{fig:year-l-memory}
    \end{subfigure}
    \caption{Setups of parameter $1/\varepsilon$ vs. maximum sketch size of \lmfd~ and \dsfd.}
    \label{fig:time-l-memory}
\end{figure}

\begin{figure}[h]
    \centering
    \begin{subfigure}[b]{0.48\linewidth}
        \centering
        \includegraphics[width=\textwidth]{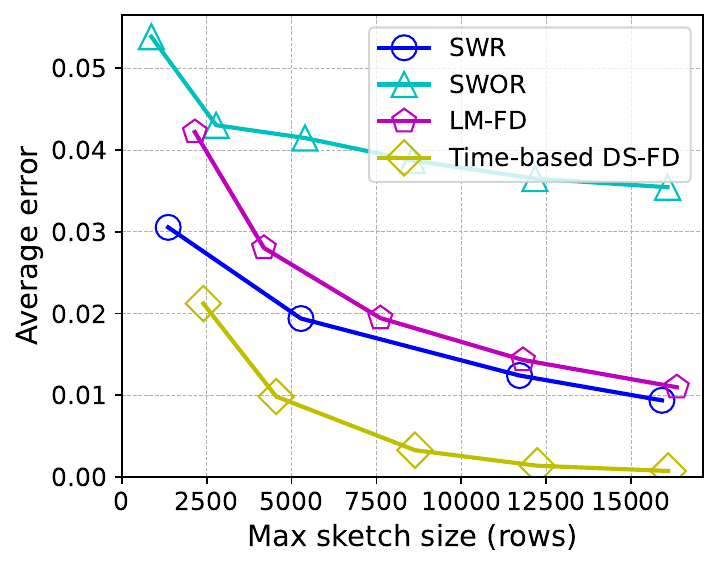}
        \caption{Average err vs. sketch size}
        \label{fig:rail-avg-err}
    \end{subfigure}
    \begin{subfigure}[b]{0.48\linewidth}
        \centering
        \includegraphics[width=\textwidth]{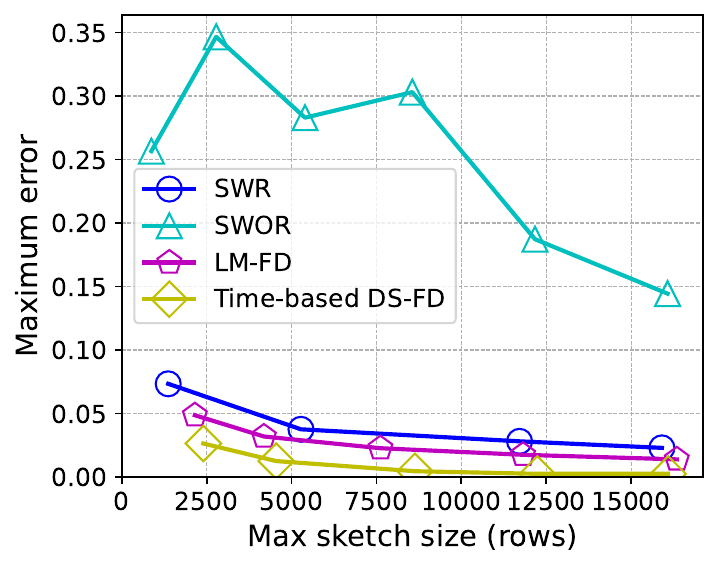}
        \caption{Maximum err vs. sketch size}
        \label{fig:rail-max-err}
    \end{subfigure}
    \caption{Error vs. sketch size on RAIL dataset.}
    \label{fig:time-avg-err}
\end{figure}

\begin{figure}[h]
    \centering
    \begin{subfigure}[b]{0.48\linewidth}
        \centering
        \includegraphics[width=\textwidth]{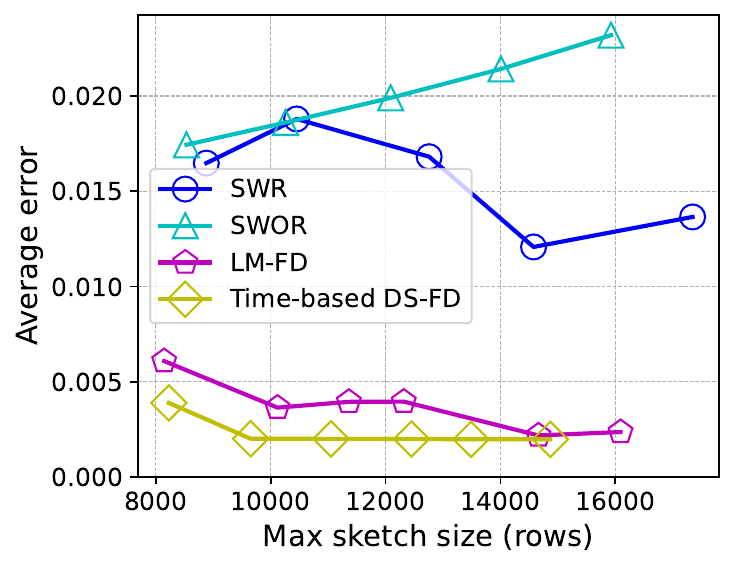}
        \caption{Average err vs. sketch size}
        \label{fig:year-avg-err}
    \end{subfigure}
    \begin{subfigure}[b]{0.48\linewidth}
        \centering
        \includegraphics[width=\textwidth]{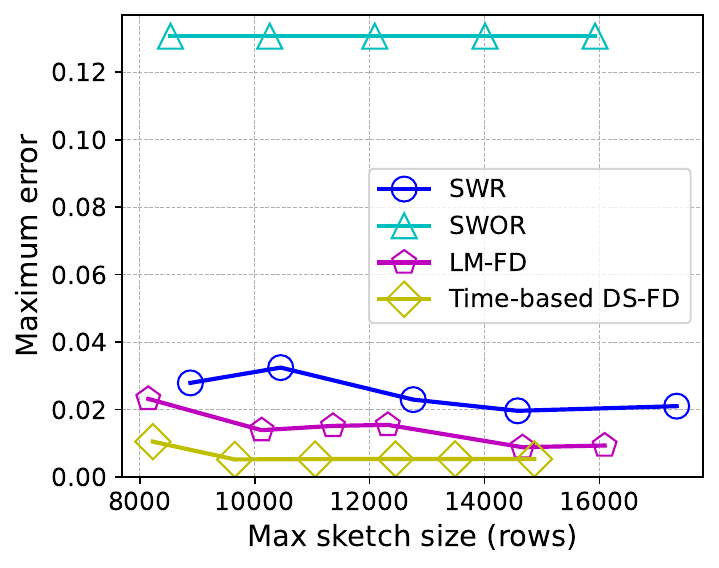}
        \caption{Maximum err vs. sketch size}
        \label{fig:year-max-err}
    \end{subfigure}
    \caption{Error vs. sketch size on YEAR dataset.}
    \label{fig:time-max-err}
\end{figure}

Subsequently, we evaluate four time-based algorithms on the two datasets tailored for time-based sliding windows. Figure~\ref{fig:time-l-memory} depicts the space overhead for \lmfd and Time-Based \dsfd with varying parameters $\ell$ on the RAIL and YEAR datasets. Additionally, Figures~\ref{fig:time-avg-err} and~\ref{fig:time-max-err} illustrate the trade-offs between the maximum sketch size and the average error, and between the maximum sketch size and the maximum error, respectively. The experiments conducted on the time-based window model yield the following observations:

(1) As shown in Figure~\ref{fig:time-l-memory}, the space cost of \lmfd escalates more rapidly than that of \dsfd as $1/\varepsilon$ increases. The actual space overheads for both \lmfd and Time-Based \dsfd align with their theoretical predictions of \(O\left(\frac{d}{\varepsilon^2} \log \varepsilon NR\right)\) and \(O\left(\frac{d}{\varepsilon} \log \varepsilon NR\right)\), respectively, corroborating the theoretical analyses.

(2) Figures~\ref{fig:time-avg-err} and~\ref{fig:time-max-err} show that Time-Based \dsfd exhibits a superior space-error trade-off compared to other algorithms on both the RAIL and YEAR datasets. This indicates that our algorithm effectively adapts to the time-based sliding window model, maintaining its performance and efficiency.

\begin{table}[ht]
\centering
\caption{Update time and query time of all methods with a relative error bound of $\varepsilon = 1/100$ on the BIBD dataset.}
\label{tab:update-time}
\sisetup{scientific-notation=true}
\begin{tabular}{|l|r|r|}
\hline
Time(\unit{\milli\second})& Update time & Query Time \\ \hline
SWR & 65.722 & 157.500 \\ \hline
SWOR & 3.143 & 291.936 \\ \hline
\lmfd & 0.061 & 3599.310 \\ \hline
\difd & 2.428 & 59.904 \\ \hline
\textbf{\dsfd} & 1.053 & 27.655 \\ \hline
\end{tabular}%

\end{table}

\header{\bf Update and query time. } We also record the average one-step update time and query time of all algorithms on BIBD dataset, as detailed in Table~\ref{tab:update-time}. Based on these observations, we draw the following conclusions:

(1) \lmfd requires the least amount of time for average updates, a finding that is in line with its update time complexity of $O\left(d \log \varepsilon NR\right)$ as reported by Wei et al.~\cite{wei2016matrix}. Conversely, the average query time for \lmfd is the highest among the evaluated methods. This increase in query time can be attributed to the time-consuming merging operation of all sketches in non-expiring blocks.

(2) Our \dsfd algorithm shows an acceptable average update time, while its average query time is the lowest among all the methods. This performance indicates that \dsfd effectively balances update and query times, making it an advantageous choice for matrix sketching over sliding windows.
\section{Conclusion}

In this paper, we delve into the challenge of matrix sketching over sliding windows, introducing a novel method, denoted as \dsfd. This method achieves space costs of \(O\left(\frac{d}{\varepsilon} \log R \right)\) and \(O\left(\frac{d}{\varepsilon} \log \varepsilon NR \right)\) for estimating covariance matrices in sequence-based and time-based sliding windows, respectively. Furthermore, we establish and validate a space lower bound for the covariance matrix estimation problem within sliding windows, demonstrating the space efficiency of our algorithm. Through extensive tests on large-scale synthetic and real-world datasets, we empirically validate the accuracy and efficiency of \dsfd, corroborating our theoretical analyses.

\begin{acks}
\begin{sloppypar}
This research was supported in part by National Science and Technology Major Project (2022ZD0114802), by National Natural Science Foundation of China (No. U2241212, No. 61932001, No. 62276066, No. 62376275), by Beijing Natural Science Foundation No. 4222028, by Beijing Outstanding Young Scientist Program (No.BJJWZYJH012019100020098), by Alibaba Group through Alibaba Innovative Research Program. We also wish to acknowledge the support provided by the fund for building world-class universities (disciplines) of Renmin University of China, by Engineering Research Center of Next-Generation Intelligent Search and Recommendation, Ministry of Education, Intelligent Social Governance Interdisciplinary Platform, Major Innovation \& Planning Interdisciplinary Platform for the “Double-First Class” Initiative, Public Policy and Decision-making Research Lab, and Public Computing Cloud, Renmin University of China.  
\end{sloppypar}
\end{acks}

\newpage
\bibliographystyle{ACM-Reference-Format}
\bibliography{sample}

\newpage
\appendix
\section{Appendix}

\subsection{Proof of Lemma~\ref{lem:fast-dsfd}}

\begin{proof}
    \begin{equation*}
        \begin{split}
            \bm{D}^\top \bm{D} &= \bm{V\Sigma}^2\bm{V}^\top\\
            \bm{D}^{\prime\top} \bm{D}^\prime&= (\bm{D}-\bm{Dv}_j\bm{v}_j^\top)^\top(\bm{D}-\bm{Dv}_j\bm{v}_j^\top)\\
            &=(\bm{D}^\top-\bm{v}_j\bm{v}_j^\top \bm{D}^\top)(\bm{D}-\bm{D}\bm{v}_j\bm{v}_j^\top)\\
            &=\bm{D}^\top \bm{D}-\bm{v}_j\bm{v}_j^\top \bm{D}^\top \bm{D}-\bm{D}^\top \bm{D} \bm{v}_j\bm{v}_j^\top+\bm{v}_j\bm{v}_j^\top \bm{D}^\top \bm{D} \bm{v}_j\bm{v}_j^\top\\
            &=(\bm{I}-\bm{v}_j \bm{v}_j^\top)\bm{D}^\top \bm{D}(\bm{I}-\bm{v}_j \bm{v}_j^\top)^\top\\
            &=(\bm{I}-\bm{v}_j \bm{v}_j^\top)V\Sigma^2V^\top(\bm{I}-\bm{v}_j \bm{v}_j^\top)^\top
        \end{split}
    \end{equation*}
\end{proof}

\subsection{Proof of Theorem~\ref{thm:norm-seq-sw-fd}}
\label{proof:norm-seq-sw-fd}

\begin{proof}
    We first prove that the \dsfd~algorithm can return a sketch that satisfies the error bound \eqref{eq:norm-seq-swfd-error}. Let $\bm{A}_{s,t}$ be the full stack of rows from time $s$ to $t$, and let $\bm{B}_{s,t}$ be the row stack of all snapshots that are dumped from time $s$ to $t$. Define $\bm{C}^\top_{s,t}\bm{C}_{s,t}$ as the covariance residual matrix, given by $\bm{C}^T_{s,t}\bm{C}_{s,t}=\bm{A}^\top_{s,t}\bm{A}_{s,t}-\bm{B}^\top_{s,t}\bm{B}_{s,t}$. Then, the estimation for $\bm{A}_{T-N, T}^\top\bm{A}_{T-N,T}$ is returned as $\bm{B}_{T-N,T}^\top\bm{B}_{T-N,T}+\hat{\bm{C}}_{(k-1)N+1,T}^\top\hat{\bm{C}}_{(k-1)N+1,T}$, where $\hat{\bm{C}}_{(k-1)N+1,T}$ is the main FD sketch at time $T$ and $k=\max(1,\lfloor(T-1)/N\rfloor)$. For clarity in the proof, $\bm{A}_{(k-1)N+1,t}$, $\bm{B}_{(k-1)N+1,t}$, $\bm{C}_{(k-1)N+1,t}$, $\hat{\bm{C}}_{(k-1)N+1,t}$ are denoted as $\bm{A}_t$, $\bm{B}_t$, $\bm{C}_t$, and $\hat{\bm{C}}_t$. The error is defined as:

    {\allowdisplaybreaks
    \begin{align*}
        &\hphantom{{}={}}\textbf{cova-err}(\bm{A}_W, \bm{B_W})\\
        &=\lVert \bm{A}_{T-N, T}^\top\bm{A}_{T-N,T}- (\bm{B}_{T-N,T}^\top\bm{B}_{T-N,T}+\hat{\bm{C}}_T^\top\hat{\bm{C}}_T)\rVert_2\\
        &=\lVert \bm{A}_T^\top\bm{A}_T -\bm{A}_{T-N}^\top\bm{A}_{T-N}-\bm{B}_T^\top\bm{B}_T+\bm{B}_{T-N}^\top\bm{B}_{T-N}-\hat{\bm{C}}_T^\top\hat{\bm{C}}_T\rVert_2\\
        &=\lVert \bm{C}_T^\top\bm{C}_T - \bm{C}_{T-N}^\top\bm{C}_{T-N}-\hat{\bm{C}}_T^\top\hat{\bm{C}}_T\rVert_2\\
        &=\lVert (\bm{C}_T^\top\bm{C}_T -\hat{\bm{C}}_T^\top\hat{\bm{C}}_T)- (\bm{C}_{T-N}^\top\bm{C}_{T-N} - \hat{\bm{C}}_{T-N}^\top\hat{\bm{C}}_{T-N})\\
        & \quad -\hat{\bm{C}}_{T-N}^\top\hat{\bm{C}}_{T-N}\rVert_2\\
        & \le \lVert \bm{C}_T^\top\bm{C}_T -\hat{\bm{C}}_T^\top\hat{\bm{C}}_T\rVert_2+ \lVert\bm{C}_{T-N}^\top\bm{C}_{T-N} - \hat{\bm{C}}_{T-N}^\top\hat{\bm{C}}_{T-N}\rVert_2\\
        &  \quad  +\lVert\hat{\bm{C}}_{T-N}^\top\hat{\bm{C}}_{T-N}\rVert_2\\
    \end{align*}
    }
    
    where $\hat{\bm{C}}_T$ and $\hat{\bm{C}}_{T-N}$ are FD sketches for $\bm{C}_T$ and $\bm{C}_{T-N}$, respectively.~\cite{liberty2022even}. In detail, we define $\bm{\Delta}_t=\bm{a}_t^\top\bm{a}_t + \hat{\bm{C}}_{t-1}^\top\hat{\bm{C}}_{t-1} - (\hat{\bm{C}}_{t}^\top\hat{\bm{C}}_{t} + \bm{b}_t^\top\bm{b}_t)=\bm{U}_t\cdot \min(\bm{\Lambda}_t, \bm{I}\cdot \lambda_\ell^t)\cdot \bm{U}_t^\top$. Then

    \begin{equation*}
        \begin{split}
        \sum_{t=(k-1)N+1}^T \bm{\Delta}_t &= \bm{A}_T^\top \bm{A}_T - \bm{B}_T^\top \bm{B}_T - \hat{\bm{C}}_{T}^\top\hat{\bm{C}}_{T}\\
        &= \bm{C}_T^\top \bm{C}_T - \hat{\bm{C}}_{T}^\top\hat{\bm{C}}_{T} \\
        \end{split}
    \end{equation*}

    Similar to the proof of the Lemma 1 in~\cite{liberty2022even}, we have

    \begin{equation*}
        \lVert \bm{C}_T^\top \bm{C}_T - \hat{\bm{C}}_{T}^\top\hat{\bm{C}}_{T} \rVert_2 \le \varepsilon \lVert \bm{C}_T \rVert_F^2
    \end{equation*}
    
    By the covariance error guarantees provided by FD, we have
    \begin{align*}
        \lVert \bm{C}_T^\top\bm{C}_T -\hat{\bm{C}}_T^\top\hat{\bm{C}}_T\rVert_2 &\le \varepsilon \lVert \bm{C}_T \rVert_F^2 \le \varepsilon \lVert \bm{A}_T \rVert_F^2\\
        \lVert\bm{C}_{T-N}^\top\bm{C}_{T-N} - \hat{\bm{C}}_{T-N}^\top\hat{\bm{C}}_{T-N}\rVert_2 &\le \varepsilon \lVert \bm{C}_{T-N} \rVert_F^2 \le \varepsilon \lVert \bm{A}_{T-N} \rVert_F^2
    \end{align*}

    Then the covariance error is bounded by
    
    \begin{equation*}
        \begin{split}
        &\textbf{cova-err}(\bm{A}_W, \bm{B}_W)\\
        \le & \varepsilon \lVert \bm{A}_T \rVert_F^2 + \varepsilon \lVert \bm{A}_{T-N} \rVert_F^2 + \varepsilon N \\
        \le & \varepsilon(T-(k-1)N) + \varepsilon (T-N-(k-1)N) + \varepsilon N \\
        =&2\varepsilon( T - (k-1)N)\le 4\varepsilon N\\
        \end{split}
    \end{equation*}
    where the last inequality holds due to the \textit{restart every $N$ step} operation.

    To prove that the algorithm takes $O\left(d/\varepsilon\right)$ memory, we need to find the upper bound on the number of snapshots in the queue $\mathcal{S}$ at any given time. Let's consider a time $T$ where the queue of snapshots $\mathcal{S}$ has $k$ snapshots of vectors $\bm{v}_1, \bm{v}_2, \dots, \bm{v}_k$. We'll analyze the case where $N > 1/\varepsilon$. Since $\lVert \bm{v}_i \rVert_2^2 \ge \varepsilon N$ for each $i$, we have $\sum_{i=1}^k \lVert \bm{v}_i \rVert_2^2 \ge k\varepsilon N$. Moreover, as the sum of energy of the snapshot vectors cannot exceed the sum of energy from the residual matrix and the input vectors, we have $\sum_{i=1}^k \lVert \bm{v}_i \rVert_2^2 \le \ell \varepsilon N+N$. Thus $2N \ge \sum_{i=1}^k \lVert \bm{v}_i \rVert_2^2 \ge k \epsilon N$, which implies $k\le 2/\varepsilon$. Finally, considering the dimension of each snapshot vector as $d$, when we sum up the space required by $\hat{\bm{C}}\in \mathbb{R}^{\ell\times d}$ and $\mathcal{S}$, we get the total memory cost of $O\left(\ell d+k d\right) = O(d/\varepsilon)$.
\end{proof}

\subsection{Proof of Theorem~\ref{thm:seq-sw-fd}}
\label{proof:thm:seq-sw-fd}

\begin{proof}
We consider the error guarantee first. For the $i$-th level of \seqdsfd, the error is $2^{i+2}\varepsilon N$ according to Theorem~\ref{thm:norm-seq-sw-fd}. To ensure that the error of \eqref{eq:error-of-seq-swfd} is bounded by $\beta \varepsilon \lVert \bm{A}_{T-N,T} \rVert_F^2$, we need to find the level $i$ such that $2^{i+2}\varepsilon N \leq \beta \varepsilon \lVert \bm{A}_{T-N,T} \rVert_F^2$. We choose the integer $i$ of 
\begin{equation}
    \label{eq:t1}
    \log_2 \frac{\beta \lVert \bm{A}_{T-N,T}\rVert_F^2}{4N}-1 < i\le \log_2 \frac{\beta \lVert \bm{A}_{T-N,T}\rVert_F^2}{4N}.
\end{equation}

Meanwhile, we must ensure that there are "enough" snapshots in level $i$. Suppose at time $T$, the queue of snapshots $\mathcal{L}[i].\mathcal{S}$ contains $k$ snapshots denoted as $\bm{v}_1, \bm{v}_2, \dots, \bm{v}_k$. Given that the dump threshold of level $i$ is $\theta = 2^i \varepsilon N$, we have $\sum_{i=1}^k \lVert \bm{v}_k \rVert_2^2 \geq k \cdot 2^i \varepsilon N$. This sum includes contributions from both the residual matrix and the Frobenius norm of $\lVert \bm{A}_W\rVert_F^2$. We can express this as an upper bound $\sum_{i=1}^k \lVert \bm{v}_k \rVert_2^2 \le \ell 2^i \varepsilon N + \lVert \bm{A}_{T-N,T}\rVert_F^2$. Thus, we obtain

\begin{equation}
    \label{eq:t2}
    2^iN+\lVert \bm{A}_{T-N,T}\rVert_F^2\ge k\cdot2^i\varepsilon N.
\end{equation}

By combining \eqref{eq:t1} and \eqref{eq:t2}, we deduce

\begin{equation}
    \left(\frac{\beta}{4}+1 \right)\lVert \bm{A}_{T-N,T}\rVert_F^2 > k\varepsilon\frac{\beta \lVert \bm{A}_{T-N,T}\rVert_F^2}{8}.
\end{equation}
which derives $k< 2\left(1+\frac{4}{\beta}\right)\frac{1}{\varepsilon}$. For instance, setting $\beta=4$, it suffices to store at most $\frac{4}{\varepsilon}$ snapshots across all levels. Therefore, the total memory complexity of our algorithm is $O\left(\frac{d}{\varepsilon}\log R\right)$.
\end{proof}

\subsection{Proof of Theorem~\ref{thm:time-swfd-lower-bound}}
\label{proof:time-swfd-lower-bound}

\begin{figure*}[h]
    \includegraphics[width=\textwidth]{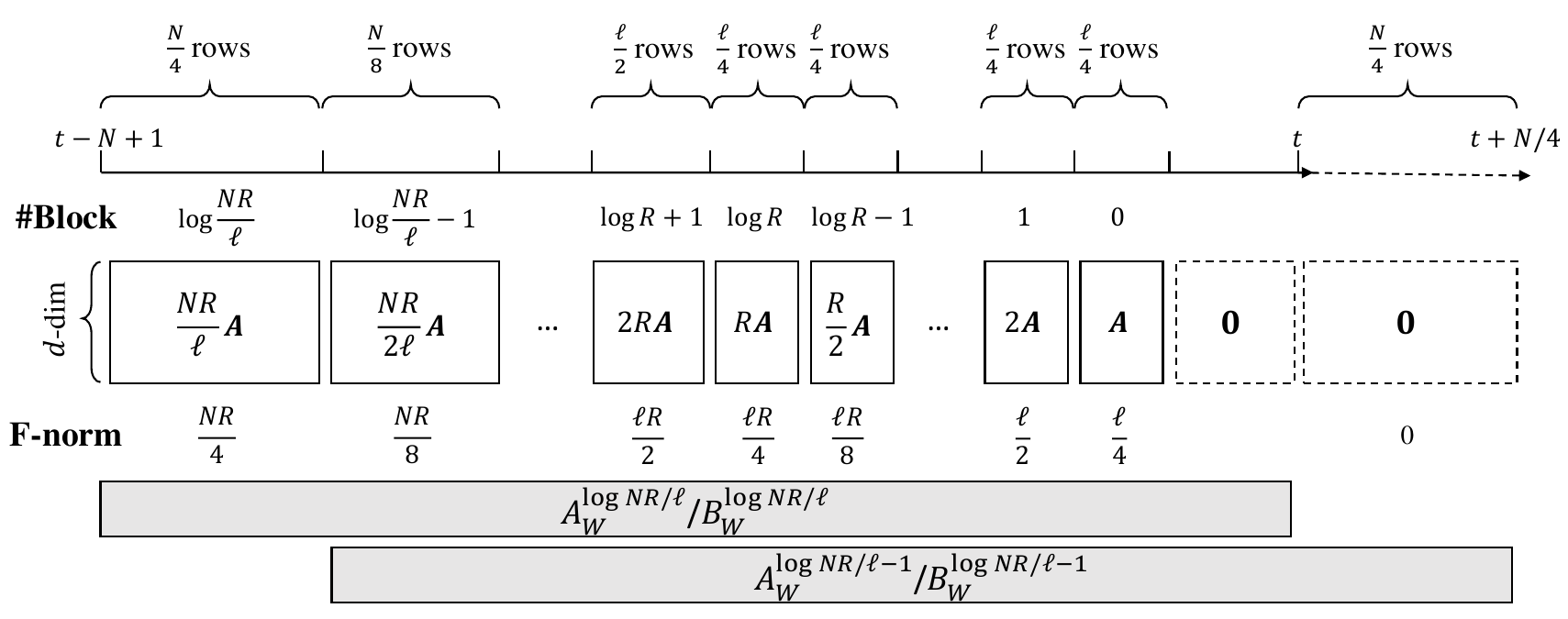}
    \caption{Proof of time-based lower bound.}
    \label{fig:time-based-lower-bound}
\end{figure*}

\begin{proof}

Similar as the proof of Theorem ~\ref{thm:seq-swfd-lower-bound}, we partition a window of size $N$ of $d$-dimension vectors into $\log NR/\ell +2$ blocks, as shown in Figure ~\ref{fig:time-based-lower-bound}, number the leftmost $\log NR/\ell +1$ blocks $\log NR/\ell, \dots, 1,0$ from left to right. We choose $\log NR/\varepsilon + 1$ matrices of $\frac{\ell}{4}\times d$ in the set of matrices $\mathcal{A}$, the total number of dinstinct arrangements is $L= {\Omega(2^{d\ell})\choose{\log NR/\ell}}$. Hence $\log L =\Omega(d\ell \log NR/\ell)$. For the block $i$, we multiply the chosen $\bm{A}_i\in \mathbb{R}^{\frac{\ell}{4}\times d}$ with a scalar of $\sqrt{2^i}$. Therefore, the square of Frobenius norm of block $i$ is $\lVert\bm{A}_i\rVert_F^2=2^i \cdot \ell/4$. For the block $i$ where $i> \log R$, we have to expand the number of rows from $\ell/4$ to $\frac{\ell}{4} \cdot 2^{i-\log R}$ to ensure that $1\le \lVert \bm{a} \rVert_2^2 \le R$. The total number of rows have to be bounded by $N$, that is, $\frac{N}{2}+\frac{\ell}{4}(\log R - 1)\le N$, which derives $N\ge \frac{\ell}{2}\log \frac{R}{2}$. 

We will assume that the algorithm is presented with one of these $L$ arrangements of length $N$, followed by a sequence of all zero vectors of length $N$. We denote $\bm{A}_W^i$ as the matrix over the sliding window of length $N$ the moment that $i+1, i+2,\dots \log R$ blocks have expired, e.g. $\bm{A}_W^{\log R}$ and $\bm{A}_W^{\log R-1}$ in Figure ~\ref{fig:time-based-lower-bound}. Suppose our sliding window algorithm gives estimations $\bm{B}_W^i$ and $\bm{B}_W^{i-1}$ of $\bm{A}_W^i$ and $\bm{A}_W^{i-1}$ with a relative error of $\frac{1}{3\ell}$, that is to say

\begin{gather*}
    \lVert {\bm{A}_W^i}^\top\bm{A}_W^i -{\bm{B}_W^i}^\top {\bm{B}_W^i}\rVert_2\le \frac{1}{3\ell} \lVert \bm{A}_W^i\rVert_F^2 = \frac{1}{3\ell}\left( \frac{\ell}{4} \cdot 2^{i+1} - \frac{\ell}{4}\right)\\
    \lVert {\bm{A}_W^{i-1}}^\top\bm{A}_W^{i-1} -{\bm{B}_W^{i-1}}^\top {\bm{B}_W^{i-1}}\rVert_2\le \frac{1}{3\ell} \lVert \bm{A}_W^{i-1}\rVert_F^2 = \frac{1}{3\ell}\left( \frac{\ell}{4} \cdot 2^{i} - \frac{\ell}{4}\right)
\end{gather*}

Then we can answer the block $i$ with $\bm{B}_i^\top \bm{B}_i = {\bm{B}_W^i}^\top {\bm{B}_W^i} - {\bm{B}_W^{i-1}}^\top {\bm{B}_W^{i-1}}$.
\begin{equation*}
    \label{eq:leftest-block-app}
    \begin{split}
        &\lVert \bm{A}_i^\top \bm{A}_i -\bm{B}_i^\top \bm{B}_i\rVert_2\\
        =&\lVert ({\bm{A}_W^i}^\top{\bm{A}_W^i} -{\bm{A}_W^{i-1}}^\top{\bm{A}_W^{i-1}}) -({\bm{B}_W^i}^\top {\bm{B}_W^i}  -{\bm{B}_W^{i-1}}^\top {\bm{B}_W^{i-1}})\rVert_2\\
        \le & \lVert {\bm{A}_W^i}^\top{\bm{A}_W^i} -{\bm{B}_W^i}^\top {\bm{B}_W^i}\rVert_2 +\lVert {\bm{A}_W^{i-1}}^\top{\bm{A}_W^{i-1}}  -{\bm{B}_W^{i-1}}^\top {\bm{B}_W^{i-1}}\rVert_2\\
        \le & \frac{1}{3\ell}\left(\frac{3}{4}\ell \cdot 2^i - \frac{\ell}{2} \right) \le \frac{1}{\ell} \lVert \bm{A}_i\rVert_F^2
    \end{split}
\end{equation*}

So the algorithm can estimate the blocks of levels of all the $\log NR/\ell$ levels of the blocks by Lemma. ~\ref{lem:fd-lower-bound}, requires $\Omega(d\ell)$ bits of space. So one of the lower bound of any deterministic algorithm is $\Omega\left(\frac{d}{\varepsilon} \log \varepsilon NR\right)$.
\end{proof}

\end{document}